\documentclass[12pt]{article}

\usepackage[left=1in,
right=1in,
top=1in,
bottom=1.2in,
footskip=.3in]{geometry}
\usepackage{url}
\usepackage{amsmath}
\usepackage{amsfonts}
\usepackage{amssymb}
\usepackage{xcolor} 
\usepackage{setspace}
\allowdisplaybreaks 
\usepackage{placeins}
\usepackage{graphicx,subcaption}
\usepackage{tabularx} 
\usepackage{booktabs}
\graphicspath{{./Examples/}} 

\usepackage{tikz}
\usetikzlibrary{graphs,graphs.standard,calc}
\usepackage{tikz-qtree}
\usepackage{hyperref}
\hypersetup{
  colorlinks,
  citecolor=blue,
  linkcolor=blue,
  urlcolor=blue}
\usepackage{algorithm}
\usepackage{algorithmicx}
\usepackage{algpseudocode}
\usepackage{diagbox}
\usepackage[round]{natbib}

\newtheorem{theorem}{Theorem}[section]

\newtheorem{condition}{Condition}[section]

\newtheorem{lemma}{Lemma}[section]
\newtheorem{proposition}{Proposition}[section]
\newtheorem{remark}{Remark}[section]
\newenvironment{proof}[1][Proof]{\noindent\textbf{#1.} }{\ \rule{0.5em}{0.5em}}

\DeclareMathOperator{\bP}{\pmb{P}}
\DeclareMathOperator{\bQ}{\pmb{Q}}
\DeclareMathOperator{\bS}{\pmb{S}}

\DeclareMathOperator{\bv}{\pmb{v}}
\DeclareMathOperator{\bV}{\pmb{V}}

\DeclareMathOperator{\bX}{\pmb{X}}
\DeclareMathOperator{\bY}{\pmb{Y}}
\DeclareMathOperator{\bZ}{\pmb{Z}}

\DeclareMathOperator{\bM}{\pmb{M}}
\DeclareMathOperator{\bN}{\pmb{N}}
\DeclareMathOperator{\bB}{\pmb{B}}
\DeclareMathOperator{\bC}{\pmb{C}}
\DeclareMathOperator{\bc}{\pmb{c}}
\DeclareMathOperator{\bD}{\pmb{D}}
\DeclareMathOperator{\bmu}{\pmb{\mu}}
\DeclareMathOperator{\bnu}{\pmb{\nu}}
\DeclareMathOperator{\b1}{\pmb{1}}
\DeclareMathOperator{\bzero}{\pmb{0}}

\newcommand\numberthis{\addtocounter{equation}{1}\tag{\theequation}}

\title{Testing for lack of fit in paired comparison data}

\author{Rahul Singh$^{1}$ and Ori Davidov$^2$}

\date{$^1${ Department of Mathematics, Indian Institute of Technology Delhi, Delhi 110016, India}\\
$^2$Department of Statistics, University of Haifa, Mount Carmel, Haifa 3498838 Israel \\
\medskip
{\small
E-mail: \texttt{wrahulsingh@gmail.com} (R Singh),  
\texttt{davidov@stat.haifa.ac.il} (O Davidov)}}

\begin{document}

\maketitle

\begin{abstract}

Linear stochastic transitivity is a central assumption in paired comparison models that is rarely verified in practice. Empirical violations, however, are common and can substantially affect inference and ranking. We develop a class of tests for detecting lack of fit in cardinal paired comparison models, where lack of fit is characterized by the presence of cyclical preferences among subsets of items. We propose a suite of tests adapted to different regimes governing the growth of the comparison graph. For a fixed number of items, the proposed procedures exhibit substantially improved power relative to the classical Kendall--Smith test and its cardinal analogue. We further extend the framework to high--dimensional, sparse comparison graphs near the connectivity threshold in random graph models. The theoretical analysis characterizes the behavior of the tests under both the null and alternative, with particular emphasis on limits of detectability and consistency. Simulation studies corroborate the theoretical findings, and applications to real data uncover substantial and previously unrecognized intransitivity and structural lack of fit.

\medskip

\end{abstract}


\section{Introduction} \label{section:introduction}

Ranking a set of items based on pairwise comparison data (PCD) \citep{David1988} arises in many settings, including political choice, marketing research, sports competitions, information retrieval, and modern online platforms. This framework has become a fundamental tool for modeling preferences and relative performance across diverse domains, e.g., \cite{Menke2008, Loewen2012, Zucco2019, Hasanuzzaman2021, Wainer2023}.

To fix ideas, suppose that there are $K$ items labeled $1,\ldots,K$ that we wish to rank. Let $Y_{ijk}$ denote the outcome of the $k^{\text{th}}$ comparison between items $i$ and $j$. The random variable (RV) $Y_{ijk}$ may be binary, ordinal, or cardinal (i.e., continuous). In this paper, we focus on cardinal PCD; the applicability of the proposed methodology to binary PCD, and in particular to Bradley--Terry models \citep{Bradley1952, Hunter2004, Cattelan2012}, is discussed in Section~\ref{section:discussion}. We assume that the observations $Y_{ijk}$ for $1\le i \ne j \le K$ and $k=1,\ldots,n_{ij}$ satisfy
\begin{equation} 
\label{model.Y_ijk}
Y_{ijk} = \nu_{ij} + \epsilon_{ijk},
\end{equation}
where $\nu_{ij}=\mathbb{E}(Y_{ijk})$ and the errors $\epsilon_{ijk}$ are independent and identically distributed RVs with zero mean and finite variance. In cardinal PCD, $Y_{ijk}=-Y_{jik}$ and $\nu_{ij}=-\nu_{ji}$ \citep{jiang2011}. Thus, the model \eqref{model.Y_ijk} is indexed by a parameter $\bnu \in \mathcal{N} \equiv \mathbb{R}^{K(K-1)/2}$, where 
\[
\bnu = (\nu_{12},\ldots,\nu_{1K},\nu_{23},\ldots,\nu_{2K}, \ldots,\nu_{K-1,K})
\]
is a vector of means arranged in lexicographic order. We refer to $\bnu$ as the preference profile governing the pairwise comparisons. Let $\mathcal{G}=(\mathcal{V},\mathcal{E})$ denote the graph whose vertices are the items $1,\ldots,K$ and whose edges are the pairs $(i,j)$ for which $n_{ij} > 0$. With each edge $(i,j)$ we associate a sample $\mathcal{Y}_{ij} = (Y_{ij1},\ldots,Y_{ijn_{ij}})$ of size $n_{ij}$; the collection of all such samples is denoted by $\mathcal{Y}$. The pair $(\mathcal{G},\mathcal{Y})$ is called a pairwise comparison graph (PCG).

It is known, see \cite{saari2014, saari2021}, that any $\bnu \in \mathcal{N}$ admits a decomposition into two orthogonal components $\bnu_{\rm linear}\in \mathcal{L}$ and $\bnu_{\rm cyclic}\in \mathcal{C}$, i.e.,
\begin{align} \label{bnu:decomposed:linear:cyclic}
\bnu=\bnu_{\rm linear}+\bnu_{\rm cyclic}.
\end{align}
Here, $\bnu_{\rm linear}$ and $\bnu_{\rm cyclic}$ capture the linear and cyclic aspects of the preference relations, respectively. The spanning sets for $\mathcal{L}$ and $\mathcal{C}$, as well as other aspects of the decomposition \eqref{bnu:decomposed:linear:cyclic}, are discussed in detail in \cite{singh2026}. Next, note that if $\bnu_{\rm cyclic}=\pmb{0}$, then
\begin{equation} \label{nu.ij=mu.i-mu.j}
\nu_{ij}=\mu _{i}-\mu _{j}
\end{equation}
for all $1\le i\ne j\le K$, where $\mu_1,\ldots,\mu_K$ are referred to as scores or merits. The relation \eqref{nu.ij=mu.i-mu.j} imposes a strong form of transitivity known as linear stochastic transitivity (LST); see \cite{Oliveira2018WST} and the references therein. In the literature on PCD, it is nearly universally assumed that \eqref{nu.ij=mu.i-mu.j} holds, thereby ruling out the possibility of genuine cyclical structure in the preference profile. Formal methods for assessing whether \eqref{nu.ij=mu.i-mu.j} holds have been lacking, despite empirical evidence indicating that such violations arise in practice.

\medskip

This paper develops tests for lack of fit (LOF) for model \eqref{model.Y_ijk} with respect to the structural assumption \eqref{nu.ij=mu.i-mu.j}, where lack of fit refers to violations of this constraint rather than to general model misspecification. We propose a suite of tests adapted to different growth regimes for the number of paired comparisons $\{n_{ij}:1\le i,j \le K\}$. The resulting methodology applies to both finite, dense PCGs and sequences of sparse PCGs in which the number of items diverges. To the best of our knowledge, LOF testing for PCD has not been studied at this level of generality. The main contributions are as follows:
\begin{enumerate}
\item For a fixed number of items, we derive the asymptotic distributions of the proposed LOF tests under the null and local alternatives (Theorems \ref{Thm-LOF.1} and \ref{Thm-LOF.2&3}) under three regimes governing the growth of the paired comparisons (Conditions \ref{Condition:all.nij}, \ref{Condition:sqrt(n)}, and \ref{Condition:sqrt(m)}), corresponding to different levels of imbalance in the PCG. We further characterize conditions under which LOF is detectable (Theorem \ref{theorem:test:powerless}).
\item For a growing number of items, we obtain the asymptotic null and alternative distributions for both fully connected graphs with $n_{ij}=m\ge 1$ and sparse graphs with $n_{ij}\le 1$ (Theorems \ref{Thm-clt-RKm}, \ref{thm:nij=1}, and \ref{thm:random:graph}). We also investigate the limits of detectability (Proposition \ref{prop 3.1}) and discuss consistency.
\item Simulation results demonstrate that the proposed tests substantially outperform existing alternatives for finite graphs, including those based on the classical \cite{kendall1940} statistic. The proposed tests also perform well in large graphs, where, to the best of our knowledge, no LOF tests are currently available. An analysis of data from the USA National Basketball Association (NBA) illustrates the persistence of cyclical structure among teams.
\end{enumerate}

The paper is organized as follows. Section \ref{section:fixed:K} develops LOF tests for finite PCGs, while Section \ref{section:infinite:K} considers the case of a growing number of items. Section \ref{sec.num.study} presents simulation results, followed by an illustrative example in Section \ref{section:illustrative:example}. Section \ref{section:discussion} concludes with a discussion, including an extension of the proposed framework to binary PCD and directions for future research.


\section{Tests for LOF in fixed comparison graphs} \label{section:fixed:K}

The least squares estimator (LSE) for model \eqref{model.Y_ijk} assuming \eqref{nu.ij=mu.i-mu.j} has been studied extensively in the literature, e.g., \cite{Mosteller1951, Kwiesielewicz1996, Csato2015}. For a comprehensive recent account, see \cite{singh2025}. The LSE solves    
\begin{equation*}  
\widehat{\bmu}=\arg \min \{Q( \bmu) : \bv^{\top}\bmu=0\} 
\end{equation*}
where the objective function is given by $Q(\bmu) =\displaystyle{\sum_{1\leq i<j\leq K}\, \sum_{k=1}^{n_{ij}}\,}(Y_{ijk}-( \mu _{i}-\mu _{j}) )^{2}$,
and typically, without any loss of generality, $\bv = \b1^{\top} = (1,\ldots,1)$. Moreover, when $\mathcal{G}$ is connected, the LSE is unique and given by
\begin{align*} 
\widehat{\bmu}= \bN^{+}\bS
\end{align*}
where $\bN^{+}$ is the Moore--Penrose inverse of $\bN$, the Laplacian of the graph $\mathcal{G}$, see \cite{bapat2010}. Here, $\bN$ is a $K \times K$ matrix with elements $\sum_{j}n_{ij}$ when $i=j$ and $-n_{ij}$ when $i \neq j$. Furthermore, $\bS=(S_{1},\ldots ,S_{K})^{\top}$ is a vector of sums where $S_{i}=\sum_{j\neq i}S_{ij}$ and $S_{ij}=\sum_{k=1}^{n_{ij}}Y_{ijk}$. The total number of paired comparisons is denoted by $n=\sum_{1\leq i<j\leq K}n_{ij}$. Throughout, we assume that: 
\begin{condition} \label{iid.errors}
The errors $\epsilon_{ijk}$ are IID with zero mean and a finite variance $\sigma^2$.  
\end{condition}
Condition \ref{iid.errors} ensures that the LSE is unbiased and has a finite variance. The assumption of homoskedasticity, in Condition \ref{iid.errors}, is easily relaxed; see Remark 2.1 in \cite{singh2025}, and therefore will not be further discussed here. 

\medskip

One way to test whether \eqref{nu.ij=mu.i-mu.j} holds is via
\begin{equation} \label{Eq.Rn}
R_{n}^{(1)}=\sum_{(i,j)\in \mathcal{E}_1}n_{ij}(\frac{{S}_{ij}}{n_{ij}}-(\widehat{\mu }_{i}-\widehat{\mu }_{j}))^{2},
\end{equation}
where $\mathcal{E}_1=\mathcal{E}$. The motivation for the use of the notation $\mathcal{E}_1$ for $\mathcal{E}$ in \eqref{Eq.Rn} will be clarified below. Note that $R_n^{(1)}$ is the squared weighted residual on $(\mathcal{G},\mathcal{Y})$. In \cite{jiang2011}, the statistic $R_{n}^{(1)}$ is referred to as a certificate of reliability without reference to any formal procedure for using it. We note that \cite{Brunelli2018} surveyed many other measures of fit for PCD; none, to the best of our knowledge, has been analyzed from a statistical perspective. 

Describing the distribution of $R_{n}^{(1)}$ requires some notation. First, let $\bD_1$ be the $|\bnu| \times |\bnu|$ diagonal matrix whose $(i,j)^{th}$ diagonal element is ${n_{ij}}$ and let $\pmb{I}$ be the $|\bnu|\times |\bnu|$ identity matrix. Next, let $\bB$ be a $\binom{K}{2}\times K$ matrix whose columns are $\pmb{b}_1,\ldots,\pmb{b}_K$. The $(i,j)^{th}$ element of $\pmb{b}_k$ is denoted  $b_{k}(i,j)$ and defined by 
\begin{align*} 
b_{k}(i,j)=\mathbb{I}(i=k) -\mathbb{I}(j=k).
\end{align*}
Observe that the matrix $\bB^{\top}$ is the incidence matrix associated with the directed, complete graph, with vertices $\mathcal{V}$. 

\begin{theorem} \label{Thm-LOF.1}
Suppose that $\bnu\in\mathcal{L}$. Assume that Condition \ref{iid.errors} holds and that the errors are IID $\mathcal{N}(0,\sigma^2)$ RVs. If so, 
\begin{equation*}
R_{n}^{(1)} \stackrel{d}{=} \sum_{1}^{r_1}\lambda_{i} Z_{i}^{2}
\end{equation*}
where $r_1=|\mathcal{E}_1|-(K-1)$, $Z_{1},\ldots ,Z_{r_1}$ are independent $\mathcal{N}(0,1)$ RVs and $\lambda _{1},\ldots ,\lambda_{r_1}$ are the non--zero eigenvalues of the $|\bnu|\times|\bnu|$ matrix
\begin{align*}
\pmb{\Omega}_1 = \sigma^2 (\bD_1^{+})^{1/2} (\pmb{I}-\bB (\bB^{\top} \bD_1\bB)^{+} \bB^{\top} \bD_1) \, \bD_1^{+} (\pmb{I}-\bB (\bB^{\top} \bD_1\bB)^{+} \bB^{\top} \bD_1)^{\top} (\bD_1^{+})^{1/2}.
\end{align*}
In addition, if $\bnu =\pmb{l}+\pmb{\delta}$ where $\pmb{l}\in\mathcal{L}$ and $\pmb{\delta}\in\mathcal{C}$ then 
\begin{equation*}
R_{n}^{(1)} \stackrel{d}{=} \sum_{1}^{r_1}\lambda_{i} (Z_{i}+\phi_i)^{2},
\end{equation*}
where $\phi_1,\ldots,\phi_{r_1}$ are the elements of the vector $\pmb{\phi}_1= \pmb{O}_1 (\pmb{\Omega}_1^{+})^{1/2}\bD_1^{1/2}\pmb{\delta}$ which correspond to the nonzero eigenvalues of $\pmb{\Omega}_1$. Here $\pmb{O}_1$ is the orthonormal matrix whose columns are the eigenvectors of $\pmb{\Omega}_1$.
\end{theorem}

If the number of paired comparisons is highly unbalanced, i.e., for some pairs $(i,j)$ and $(i',j')$ in $\mathcal{E}$ we have $n_{ij}\gg n_{i'j'}$, then although each comparison contributes (approximately) one degree of freedom, as reflected by $r_1$, pairs with large sample sizes effectively drive the power of \eqref{Eq.Rn}. Moreover, when the errors are not normally distributed, Theorem \ref{Thm-LOF.1} does not apply, since $\sqrt{n_{ij}}\{S_{ij}/n_{ij} - (\widehat{\mu}_i - \widehat{\mu}_j)\}$ is approximately normal only when $n_{ij}$ is large; for small or moderate sample sizes, the summands in \eqref{Eq.Rn} remain non-negligibly non--Gaussian, precluding a simple closed--form limit for $R_n^{(1)}$. Accordingly, we describe the limiting distributions of two modifications of \eqref{Eq.Rn}, denoted $R_n^{(2)}$ and $R_n^{(3)}$, which address these issues and apply to all error distributions satisfying Condition~\ref{iid.errors}. First, however, it is important to note that, until recently, the prevailing assumption in the literature on PCD was that:

\begin{condition} \label{Condition:all.nij}
For all edges $(i,j)\in \mathcal{E}$ we have $n_{ij}/n\to \theta_{ij} \in (0,1)$ as $n\to \infty$.
\end{condition}

In other words, the number of paired comparisons grows at the same rate across all edges. Under Condition~\ref{Condition:all.nij}, Theorem~\ref{Thm-LOF.1} continues to hold in the limit even when the errors are not normal RVs. However, less restrictive asymptotic regimes also guarantee consistency and asymptotic normality of the LSE, thereby yielding limiting distributions for the corresponding LOF tests. Specifically, \cite{singh2025} provides graph--based necessary and sufficient conditions for consistency and asymptotic normality, which we list here for completeness:

\begin{condition} \label{Condition:sqrt(m)}
There exists a spanning tree $\mathcal{T}\subset\mathcal{G}$ such that $\min \{n_{ij}:( i,j) \in \mathcal{T} \}\rightarrow \infty $ as $n\rightarrow \infty$.
\end{condition}
Condition \ref{Condition:sqrt(m)} states that, on at least one spanning tree, $\mathcal{T}\subset\mathcal{G}$ sample sizes increase to infinity. Thus, guaranteeing asymptotic normality at a rate of $\sqrt{m}$ where 
\begin{align*} 
m =\displaystyle{\max_{\mathcal{T}\subset \mathcal{G}}}\min \{n_{ij}:(i,j)\in \mathcal{T} \}.
\end{align*}
Although $m\to\infty$ as $n \to \infty$ the ratio $m/n$ can be arbitrarily small. 

\begin{condition} \label{Condition:sqrt(n)}
There exists a spanning tree $\mathcal{T}\subset \mathcal{G}$ such that
\begin{equation*}
\min \{n_{ij}:(i,j)\in \mathcal{T}\}/n\to c\in(0,\infty)\text{ as }n\to\infty.
\end{equation*}
\end{condition}
Condition \ref{Condition:sqrt(n)} results in  $\sqrt{n}$ scaling and  implies Condition \ref{Condition:sqrt(m)}. Also observe that when Condition \ref{Condition:sqrt(n)} holds, then $m/n$ converges to a positive constant. For more details, see Theorems 3.3 and 3.4 in \cite{singh2025}. 

\medskip

We now introduce two modifications of \eqref{Eq.Rn}, denoted $R_{n}^{(s)}$ for $s \in \{2,3\}$, adapted to Conditions \ref{Condition:sqrt(m)} and \ref{Condition:sqrt(n)}, respectively. The key difference is that, instead of summing over all edges in $\mathcal{E}_1=\mathcal{E}$ as in \eqref{Eq.Rn}, $R_{n}^{(2)}$ sums only over pairs $(i,j) \in \mathcal{E}_2$ while $R_{n}^{(3)}$ sums over $(i,j) \in \mathcal{E}_3$ where $\mathcal{E}_2 = \{(i,j)\in \mathcal{E}:\, n_{ij}\to\infty\}$ and $\mathcal{E}_3 = \{(i,j)\in \mathcal{E}:\, n_{ij}/n\to c>0 \text{ as } n\to\infty\}
$. Let $\bD_2$ be the $|\bnu|\times |\bnu|$ diagonal matrix whose $(i,j)^{th}$ diagonal entry is $n_{ij}$ if $(i,j)\in\mathcal{E}_2$ and $0$ otherwise. Similarly, let $\bD_3$ be the $|\bnu|\times |\bnu|$ diagonal matrix whose $(i,j)^{th}$ diagonal entry is $n_{ij}$ if $(i,j)\in\mathcal{E}_3$ and $0$ otherwise. Finally, for $s\in\{2,3\}$ define
\[
\pmb{Q}_s=\bD_s^{1/2}(\pmb{I} - \bB (\bB^{\top} \bD_1\bB)^{+} \bB^{\top} \bD_1)\;(\bD_1^{+})^{1/2},
\]
and let $\pmb{\Omega}_s=\sigma^2\lim_{n\to\infty} \pmb{Q}_s\;\pmb{Q}_s^{\top}$.

\begin{theorem} \label{Thm-LOF.2&3}
Assume $\bnu\in\mathcal{L}$ and Condition \ref{iid.errors} holds. If Condition \ref{Condition:sqrt(m)} holds then $\pmb{\Omega}_2$ exists and 
\begin{equation*}
R_{n}^{(2)}\Rightarrow \sum_{1}^{r_2}\lambda_{i} Z_{i}^{2}
\end{equation*}
as $n \to \infty$, where $\lambda _{1},\ldots ,\lambda_{r_2}$ are the non--zero eigenvalues of $\pmb{\Omega}_2$ and $r_2={\rm rank}(\pmb{\Omega}_2)$. In addition, if Condition \ref{Condition:sqrt(n)} holds then $\pmb{\Omega}_3$ exists and 
\begin{equation*}
R_{n}^{(3)}\Rightarrow \sum_{1}^{r_3}\lambda_{i} Z_{i}^{2}
\end{equation*}
as $n \to \infty$, where $\lambda _{1},\ldots ,\lambda_{r_3}$ are the non--zero eigenvalues of $\pmb{\Omega}_3$ and $r_3={\rm rank}(\pmb{\Omega}_3)$. 

Next, suppose $\bnu_n=\pmb{l}+n^{-1/2}\pmb{\delta}$ where $\pmb{l}\in\mathcal{L}$ and $\pmb{\delta}\in \mathcal{C}$ is a fixed vector. Then under the Conditions \ref{Condition:sqrt(m)} and \ref{Condition:sqrt(n)} respectively for $s\in\{2,3\}$,
\begin{align*}
R_{n}^{(s)} \Rightarrow \sum_{i=1}^{r_s}\lambda
_{i}(Z_{i}+\phi_i)^{2}, 
\end{align*}
as $n\rightarrow \infty$ and $\phi_1,\ldots,\phi_{r_i}$ are the elements of the vector $\pmb{\phi}_s= \pmb{O}_s (\pmb{\Omega }_s^{+})^{1/2}\pmb{\Xi}_s^{1/2}\pmb{\delta}$ which correspond to the nonzero eigenvalues of $\pmb{\Omega}_s$. Here $\pmb{O}_s$ is the orthonormal matrix whose columns are the eigenvectors of $\pmb{\Omega}_s$ and $\pmb{\Xi}_s=\lim \bD_s/n$. 
\end{theorem}

Theorems \ref{Thm-LOF.1} and \ref{Thm-LOF.2&3} provide three methods for assessing the fit of model \eqref{model.Y_ijk} under the assumption \eqref{nu.ij=mu.i-mu.j}, under different conditions on the error distribution and the growth of the PCG as $n\to\infty$. Numerical comparisons are presented in Section~\ref{sec.num.study}. Statements about the rank and structure of $\pmb{\Omega}_s$, $s \in \{2,3\}$, require specifying how the number of paired comparisons $\{n_{ij} : (i,j) \in \mathcal{E}_s\}$ grows with $n$. We next examine the power of the tests $R_n^{(s)}$, $s\in\{1,2,3\}$, to detect LOF.

\begin{theorem} \label{theorem:test:powerless} 
If $\mathcal{E}_s$, where $s\in\{1,2,3\}$, is a tree in $\mathcal{G}$ then $R_n^{(s)} = 0$ and therefore the test $R_n^{(s)}$ can not detect a LOF. Furthermore, if for all triads $(i,j,k)$ satisfying 
\begin{align} \label{eq:linear:nu:restriction}
\nu_{ij}+\nu_{jk}+\nu_{ki} \neq 0,
\end{align}
$(i,j),(i,k),(j,k)\notin\mathcal{E}_s$ then are $\pmb{\phi}_s=0$. In other words, the power of the tests  $R_n^{(s)}, s\in\{1,2,3\}$ does not exceed $\alpha$ and, therefore, the tests are not consistent.  
\end{theorem}

Theorem \ref{theorem:test:powerless} shows that the tests $R_n^{(s)}$, $s\in\{1,2,3\}$, can detect a departure from the null if and only if $\mathcal{E}_s$ contains a triad $(i,j,k)$ satisfying \eqref{eq:linear:nu:restriction}. Such triads are referred to as inconsistent, or cyclical, triads; see \cite{jiang2011, saari2021, singh2026} for further details.

The tests $R_n^{(1)}, R_n^{(2)}, R_n^{(3)}$ are straightforward to implement. Compare their realized values, denoted $r_n^{(1)}, r_n^{(2)}, r_n^{(3)}$, with their respective $\alpha$-level critical values obtained by simulating the null distributions given in Theorems \ref{Thm-LOF.1} and \ref{Thm-LOF.2&3}. These distributions are invariant with respect to $\bnu \in \mathcal{L}$, in which case $\bnu = \bB \bmu$; hence the particular choice of $\bmu \in \mathbb{R}^K$ used to obtain the null distribution is immaterial.

\section{Tests for LOF in growing comparison graphs}
\label{section:infinite:K}

Testing $H_0:\bnu \in \mathcal{L}$ versus $H_1:\bnu \notin \mathcal{L}$ is equivalent to testing
\begin{align} \label{LOF:K-infty}
H_0: \|\bnu_{\rm cyclic}\|=0 \quad \text{against} \quad H_1: \|\bnu_{\rm cyclic}\|\neq 0.
\end{align}
Form \eqref{LOF:K-infty} is convenient, since the power of the proposed tests is explicitly governed by $\|\bnu_{\rm cyclic}\|$. A natural generalization of the test statistic \eqref{Eq.Rn}, suitable when $K\to \infty$, is
\begin{align}\label{eq:Rn-modified}
N_{K}^{-1}\sum_{1\leq i<j\leq K} n_{ij}(\widehat{\nu}_{ij}-(\widehat{\mu }_{i}-\widehat{\mu }_{j}))^{2},
\end{align}
where $N_{K}=\sum_{1\leq i<j\leq K} \mathbb{I}_{(n_{ij}\neq 0)}$ is the number of pairs compared, and $\widehat{\nu}_{ij}$ is the $(i,j)^{th}$ element of $\widehat{\bnu}$, with $\widehat{\bnu} = \arg\min_{\bnu\in\mathcal{N}}\sum_{1\leq i<j\leq K}\sum_{k=1}^{m} (Y_{ijk}- \nu_{ij})^{2}$.

When $K\to\infty$, the growth regime governing the number of paired comparisons, i.e., the collection $\{n_{ij} : (i,j) \in \mathcal{E}_K\}$, is highly flexible. For simplicity, we first assume that $n_{ij}=m$ for all pairs $(i,j)$, so that \eqref{eq:Rn-modified} becomes
\begin{align}\label{eq:Rmk}
R_{m,K}=\binom{K}{2}^{-1}\sum_{1\leq i<j\leq K} (\widehat{\nu}_{ij}-(\widehat{\mu }_{i}-\widehat{\mu }_{j}))^{2}.
\end{align}
If $m \to \infty$ faster than $K$, then the results of Section \ref{section:fixed:K} remain valid. We therefore focus on the more realistic regime in which $m \ll K$, even when $m \to \infty$. The large--sample properties of \eqref{eq:Rmk} are described below.

\begin{condition} \label{4th.moment}
The errors satisfy $\mathbb{E}(\epsilon_{ijk}^4) <\infty$.
\end{condition}

\begin{theorem} \label{Thm-clt-RKm}
Suppose Condition \ref{iid.errors} holds. Then, for any $m$, and as $K\to\infty$, 
\begin{align*}
R_{m,K} \xrightarrow[]{p} \lim \, \psi_K^2+ \sigma^2/m,
\end{align*}
where $\psi_K^2=\, \binom{K}{2}^{-1}\|{\bnu}_{\rm cyclic}\|^2$. If Condition \ref{4th.moment} holds then 
\begin{align*}
\frac{R_{m,K} -  (\psi_K^2+ \sigma^2/m)}{\sqrt{{\rm Var}(R_{m,K})}} \Rightarrow \mathcal{N}(0, 1).
\end{align*}
Moreover, if errors are normally distributed, then ${\rm Var}(R_{m,K})$ simplifies to 
\begin{align}\label{eq:var:RK}
2\binom{K}{2}^{-2}\binom{K-1}{2}\frac{\sigma^4}{m^2} +4\,\binom{K}{2}^{-1}  \psi_K^2\,\frac{\sigma^2}{m},
\end{align}
and the approximate power function is 
\begin{align}\label{eq:power:RK}
1-\Phi(z_{1-\alpha}-\frac{\psi_K^2}{
\sqrt{2\binom{K}{2}^{-2}\binom{K-1}{2}\frac{\sigma^4}{m^2}}}).
\end{align}
where $\Phi$ is the CDF and $Z_{1-\alpha}$ is the $(1-\alpha)^{th}$-quantile of the standard normal distribution.  
\end{theorem}

Theorem \ref{Thm-clt-RKm} characterizes the limiting behavior of $R_{m,K}$. In particular, if $\psi_K^2 = 0$, then $R_{m,K}$ converges to $\sigma^2/m$ under the null, and the same limit holds along sequences for which $\psi_K^2 \to 0$. Consequently, such alternatives are not distinguishable from the null. More generally, the test has nontrivial power if and only if $K\psi_K^2 \to c$ for some $c>0$; otherwise it cannot detect LOF. In particular, the test is consistent when $K\psi_K^2 \to \infty$.

In general, ${\rm Var}(R_{m,K})$ depends on the first four moments of the error distribution and does not admit a closed-form expression. It can, however, be consistently estimated using resampling methods such as model-based resampling for linear models; see \cite{Davison1997} for details. Consequently, testing \eqref{LOF:K-infty} using $R_{m,K}$ is straightforward. In particular, when the errors are normally distributed, ${\rm Var}(R_{m,K})$ admits a closed--form expression that is readily estimable. Since
\begin{align*}
S^2_{m,K}=\frac{1}{m-1}\binom{K}{2}^{-1}\sum_{1\leq i<j\leq K}\sum_{k=1}^m (Y_{ijk}-\widehat{\nu }_{ij})^{2} \xrightarrow{p}\sigma^2,
\end{align*}
as $K \to \infty$ for all $m>1$, under both the null and alternative, a plug-in estimator is available.

Proposition \ref{prop 3.1} below provides necessary and sufficient conditions under which $K\psi_K^2 \to 0$. Its statement requires some additional notation. Following \cite{singh2026}, write $\bnu_{\rm cyclic}=\bC\pmb{\gamma}$, where $\bC$ is a $\binom{K}{2}\times \binom{K}{3}$ matrix whose columns $\{\pmb{c}_{(i,j,k)}\}_{1\leq i< j<k\leq K}$ are defined by 
\begin{align*} 
c_{(i,j,k)}(s,t) = \mathbb{I}((s,t)\in\{(i,j),(j,k),(k,i)\}) - \mathbb{I}((s,t)\in\{(j,i),(k,j),(i,k)\}). 
\end{align*}
The vector $\pmb{c}_{(i,j,k)} \in \mathcal{C}$ encodes a cyclical relation among the items in the triad $(i,j,k)$. Since $\dim(\mathcal{C})=\binom{K-1}{2}$, the representation $\pmb{\gamma}$ is not unique. Let $\Gamma_{\bnu_{\rm cyclic}} = \{\pmb{\gamma}\in\mathbb{R}^{\binom{K}{3}}:\, \bnu_{\rm cyclic}=\bC\pmb{\gamma} \}$. For each $\pmb{\gamma}\in \Gamma_{\bnu_{\rm cyclic}}$, let $s(\pmb{\gamma}) = \sum \mathbb{I}(\gamma_{ijk}\neq 0)$ denote the model size. A representation $\pmb{\gamma}\in\Gamma_{\bnu_{\rm cyclic}}$ is minimal if $s(\pmb{\gamma})\leq s(\pmb{\gamma}')$ for all $\pmb{\gamma}'\in \Gamma_{\bnu_{\rm cyclic}}$. Minimal representations are typically unique; denote them by $\pmb{\gamma}^{\rm min}$. Then
\begin{align}\label{eq:prop 3.1}
\|\bnu_{\rm cyclic}\|^2 
= \|\sum_{1\leq i<j<k\leq K} \gamma_{ijk}^{\rm min} \, \pmb{c}_{(i,j,k)}\|^2 
\le 3 \sum_{1\leq i<j<k\leq K} (\gamma_{ijk}^{\rm min})^2 
\le 3s (\gamma_{*}^{\rm min})^2,
\end{align}
where $s=s(\pmb{\gamma}^{\rm min})$ and $\gamma_{*}^{\rm min}= \max |\gamma_{ijk}^{\rm min}|$. The dependence of these quantities on $K$ is suppressed.

\begin{proposition}\label{prop 3.1}
If $\gamma_{*}^{\rm min}<B<\infty$, then $K\psi_K^2 \to 0$ if and only if
\begin{align*}
\frac{s}{K} \to 0.     
\end{align*} 
\end{proposition}

Proposition \ref{prop 3.1} implies that if the number of cyclical triads grows at a slower rate than the number of items, then departures from the null are not detectable. In particular, if $m>1$ and $K\psi_K^2>0$, then Theorem \ref{Thm-clt-RKm} can be used to test for LOF. The proof of Proposition \ref{prop 3.1} follows from \eqref{eq:prop 3.1} and is therefore omitted.

\medskip

Testing for LOF when $m=1$ and/or $K\psi_K^2 \to 0$ requires localization, i.e., that only a small fraction of items is cyclical, which is a reasonable assumption in many applications. Let $(\mathcal{G}_1,\mathcal{Y}_1),(\mathcal{G}_2,\mathcal{Y}_2),\ldots$ be a sequence of increasing PCGs, and suppose that for all sufficiently large $K$ there exists a subset $\mathcal{U}_K \subset \mathcal{V}_K$ of size $J=J_K$ such that:

\begin{condition}\label{condition:U,V}
The items in $\mathcal{U}_{K}$ are suspected of violating \eqref{nu.ij=mu.i-mu.j} whereas the items in $\mathcal{V}_K\setminus\mathcal{U}_K$ do not. Furthermore $J \to \infty$ and $J=o(K)$.
\end{condition}

Condition \ref{condition:U,V} is further discussed in Section \ref{section:discussion}. For simplicity, the dependence of $\mathcal{U}$ and $\mathcal{V}$ on $K$ is suppressed whenever possible. Additionally, and without any loss of generality, $\mathcal{U}=\{1,\ldots, J\}$ is assumed, so $\mathcal{V}\setminus\mathcal{U}=\{J+1,\ldots, K\}$. Denoting $R_{1,K}$ and $R_{1,J}$ by $R_K$ and $R_{J}$, respectively, we have:

\begin{theorem} \label{thm:nij=1}
Suppose that Conditions \ref{iid.errors}, \ref{4th.moment} and \ref{condition:U,V} hold. Then,  
\begin{align}\label{eq:thm3.2}
\frac{R_J-R_K-\psi_J^2}{\sqrt{{\rm Var}({R_J})}}\Rightarrow \mathcal{N}(0,1).
\end{align}
where $\psi_J^2= \binom{J}{2}^{-1}  \|\bnu_{J,\,\rm cyclic}\|^2$. If the errors are normally distributed, ${\rm Var}(R_{J})$ and the approximate power function reduce to \eqref{eq:var:RK} and \eqref{eq:power:RK}, respectively, with $K$ replaced by $J$ and $m$ replaced by $1$.
\end{theorem}

Localizing the search for LOF to a subset $\mathcal{U} \subset \mathcal{V}$ serves two purposes. First, it permits identification and estimation of $\sigma^2$ using the items in $\mathcal{V}\setminus \mathcal{U}$; without this separation, $\sigma^2$ is not identifiable, and testing for LOF is not possible. Second, it concentrates the test on departures from the null within $\mathcal{U}$, enabling detection when $\lim J\psi_J^2 > 0$ even if $\lim  K\psi_K^2 = 0$.  

\medskip

Theorem \ref{thm:nij=1} can be extended to sparse graphs, i.e., to PCGs in which $n_{ij}=0$ holds for a large proportion of pairs $(i,j)$ and $n_{ij}=1$ holds for the remaining pairs. Let $b_{ij}=\mathbb{I}_{\{(i,j)\in\mathcal{E}_K\}}$ indicate whether $(i,j)$ is an edge in the $K^{th}$ comparison graph. Define:
\begin{align}\label{eq:RjRk:reduced}
R_{J}'=~&\frac{\sum_{(i,j)\in\,\mathcal{U}_K} b_{ij}(Y_{ij}-(\widehat{\mu }_{i,J}-\widehat{\mu }_{j,J}))^{2}}{\sum_{(i,j)\in\,\mathcal{U}_K}b_{ij}}, ~ 
R_{K}'=~\frac{\sum_{(i,j)\in\mathcal{V}_K} b_{ij}(Y_{ij}-(\widehat{\mu }_{i}-\widehat{\mu }_{j}))^{2}}{\sum_{(i,j)\in\,\mathcal{V}_K}b_{ij}},
\end{align}
where $\psi_J^{'2}=~(\sum_{(i,j)\in\,\mathcal{V}_K}b_{ij})^{-1}\sum_{(i,j)\in\,\mathcal{U}_K}b_{ij}\nu_{ij,\,\rm cyclic}$ and $\widehat{\mu}_{1,J},\ldots, \widehat{\mu}_{J,J}$ are computed on $\mathcal{U}$. A modification of Theorem \ref{thm:nij=1} holds provided the set of indicators $\{b_{ij}\}_{1\le i <j \le K}$ satisfies $\sum_{i:(i,j)\in\,\mathcal{U}_K} b_{ij} \to \infty$ and $\sum_{i:(i,j)\in\,\mathcal{V}_K} b_{ij} \to \infty$ for every $j$. In addition, it is required that 
\begin{align}\label{eq:diverge:triads}
\sum_{(i,j),(i,k),(j,k)\in\,\mathcal{U}_K} \mathbb{I}(b_{ij}+b_{ik}+b_{jk}=3) \to \infty,
\end{align}
where $\mathbb{I}(b_{ij}+b_{ik}+b_{jk}=3)=1$ only if the triad $(i,j,k)$ is connected. Condition \eqref{eq:diverge:triads} is necessary but not sufficient for $\lim  J\psi_J^2 \to c> 0$; otherwise LOF is not detectable. 

\bigskip

The foregoing discussion naturally leads to an extension to random graphs, which capture large--scale comparison systems in which interactions arise randomly, such as online platforms or crowdsourced environments. Recently, several authors have investigated such large random PCGs, e.g., \cite{han2020,gao2023}. In particular, if we replace the constants $\{b_{ij}:\, 1\leq i<j\leq K\}$ with the RVs $\{B_{ij}:\, 1\leq i<j\leq K\}$, where $B_{ij}$ are Bernoulli$(p_K)$, then the resulting graph is an Erdős–Rényi random graph with parameters $(K,p_K)$. Let $\widetilde{R}_J$, $\widetilde{R}_K$, and $\widetilde{\psi}_J^{2}$ denote the modifications of \eqref{eq:RjRk:reduced} obtained by replacing $b_{ij}$ with $B_{ij}$.

\begin{theorem}\label{thm:random:graph}
Suppose that Conditions \ref{iid.errors}, \ref{4th.moment} and \ref{condition:U,V} hold. Assume further that $p_K\geq c\,(\log J)^{1+\epsilon}/J$ for some $c,\epsilon>0$. Then as $K\to\infty$,  $\widetilde{\psi}_J^2/\psi_J^2 \xrightarrow{p} 1$ and 
\begin{align}\label{eq:lof:random}
\frac{\widetilde{R}_J-\widetilde{R}_K-\widetilde{\psi}_J^2}{\sqrt{{\rm Var}({\widetilde{R}_J})}}\Rightarrow \mathcal{N}(0,1).
\end{align}
\end{theorem}

Theorem~\ref{thm:random:graph} shows that LOF can be detected even on very sparse graphs, a somewhat unexpected phenomenon, provided the subset $\mathcal{U}$ on which departures from \eqref{nu.ij=mu.i-mu.j} may occur is known. The condition $p_K\geq c\,(\log J)^{1+\epsilon}/J$ ensures that $\mathcal{U}$ is connected and that the number of triangles within $\mathcal{U}$ diverges \citep{Janson2011}, which is a stochastic analogue of \eqref{eq:diverge:triads}. Since $\widetilde{\psi}_J^2/\psi_J^2 \xrightarrow{p} 1$, detection is possible whenever $J\psi_J^2 \to c>0$. By Proposition~\ref{prop 3.1}, this occurs only if $s = O(J)$. In contrast with Theorems \ref{Thm-clt-RKm} and \ref{thm:nij=1}, the expression for ${\rm Var}(\widetilde{R}_J)$ does not simplify even when the errors are normally distributed. This is because the quadratic form is weighted by independent Bernoulli$(p_K)$ variables, so the effective error term is the product of a $\mathcal{N}(0,\sigma^2)$ random variable and a Bernoulli$(p_K)$ random variable, and is therefore no longer normally distributed.

\section{Simulation study} \label{sec.num.study}

We consider both finite and large comparison graphs.

\medskip

\subsection{Finite comparison graphs}

Section \ref{section:kend-smith} compares the proposed methodology to a classical alternative, namely the test of \cite{kendall1940}. In Section \ref{section:f-test} it is shown that the proposed methodology has adequate power relative to a benchmark test which uses explicit information on $\bnu_{\rm cyclic}$. Sections \ref{section:mod-unbalanced-graphs} and \ref{section:high-unbalanced-graphs} explore the behavior of the proposed methods on moderately and highly unbalanced comparison graphs where the statistics $R_n^{(2)}$ and $R_n^{(3)}$ play an important role.

\subsubsection{Comparing $R_n^{(1)}$ with the Kendall--Smith test}
\label{section:kend-smith}

We begin by comparing the proposed procedure with the classical test of \cite{kendall1940}. Although originally developed for binary PCD, the Kendall--Smith test applies more broadly and remains a standard tool for assessing LOF; see, e.g., \cite{Kingsley2010, Romdhani2014, Laird2018}. The \cite{kendall1940} statistic counts the number of cyclic triads, that is,
$$T_{{\rm KS}} = \sum_{1 \le i < j < k \le K} T_{ijk}$$
where $T_{ijk}=1$ if the triad $(i,j,k)$ is cyclical. Operationally, let $W_{ij}=\sum_{k=1}^{n_{ij}}\mathbb{I}(Y_{ijk}>0),\ Z_{ij}=\mathbb{I}(W_{ij}\ge n_{ij}/2)$ and set
\[
T_{ijk}= Z_{ij} Z_{jk} Z_{ki} + Z_{ik} Z_{kj} Z_{ji}.
\]
Then $T_{ijk}=1$ if either $Z_{ij}=Z_{jk}=Z_{ki}=1$ or $Z_{ik}=Z_{kj}=Z_{ji}=0$. In particular, $Z_{ij}=Z_{jk}=Z_{ki} = 1$ corresponds to the cycle $i \succ j \succ k \succ i$, where $\succ$ denotes the usual preference relation. The second term in the display above is associated with the cycle $i \succ k \succ j \succ i$. 

The approach of Kendall and Smith can be naturally generalized to cardinal PCD. Assume for simplicity that $n_{ij}=m$ for all $(i,j)$ and for each triad $(i,j,k)$, define the cyclic sum  
$$S^{\rm card}_{ijk} = \bar S_{ij} + \bar S_{jk} + \bar S_{ki}.$$
Here ${\rm Var}(S^{card}_{ijk})=3\sigma^2/m$. When \eqref{nu.ij=mu.i-mu.j} holds then $\mathbb{E}(S^{card}_{ijk}) = 0$, however $\mathbb{E}(S^{card}_{ijk}) \neq 0$ indicates the presence of a cyclic component. Assuming normality of errors, at $5\%$ level, a natural cardinal analogue of the Kendall-Smith test statistic is
$$T_{\mathrm{KS}}^{\mathrm{card}} = \sum_{1 \le i < j < k \le K}T_{ijk}^{\rm card}, $$
where $T_{ijk}^{\rm card} = \mathbb{I}( \sqrt{m}\,|S^{card}_{ijk}|/\sqrt{3\,\widehat{\sigma}^2}\, > 1.96 )$. 
Table \ref{table:finite:compare-all} presents a comparison of the power of $T_{{\rm KS}}$,  $T_{{\rm KS}}^{{\rm card}}$ and $R_n^{(1)}$. The power of the tests  $T_{{\rm KS}}$ and $T_{{\rm KS}}^{{\rm card}}$ was calculated by first computing their critical values under the null and then estimating the rejection probabilities under the alternative. 

\begin{table}[h]
\centering
\caption{Estimated power of the tests $T_{{\rm KS}}$, $T_{\mathrm{KS}}^{\mathrm{card}}$ and $R_n^{(1)}$ when $K=30$, $m=10$, $\bnu_{\rm cyclic}=\bc_{(1,2,3)}+\bc_{(1,2,4)}$ and normal errors. Critical values and rejection probabilities are computed using $10^5$ simulation runs.} 

\label{table:finite:compare-all}
\begin{tabular}{l|cccccc} \hline
 Test statistic & $T_{{\rm KS}}$ & & $T_{\mathrm{KS}}^{\mathrm{card}}$ &  & $R_n^{(1)}$  &   \\ \hline
 Empirical power & 0.078 & & 0.496 &  & 0.818 &   \\ \hline
\end{tabular} 
\end{table}

Table \ref{table:finite:compare-all} shows that the Kendall--Smith test has limited power for detecting LOF. Incorporating cardinal information improves performance, but the resulting procedure remains only moderately powerful. This reflects a structural limitation of the Kendall--Smith approach, which relies on local triad--based discrepancies. In contrast, the proposed approach attains substantially higher power by aggregating information across the full collection of paired comparisons. 

\subsubsection{Comparing $R_n^{(1)}$ with the regression based $F$-test}
\label{section:f-test}

When the minimal model is known and the errors are Gaussian with constant variance, the regression--based $F$ statistic coincides with the likelihood ratio test. In this setting, the $F$--test is uniformly most powerful invariant (UMPI) under the natural group of orthogonal transformations and is also most powerful among unbiased tests; see \cite{Lehmann2005}. It therefore provides a natural benchmark for evaluating the performance of $R_n^{(1)}$. To fix ideas, suppose that the minimal model is of the form $\bnu_{\rm cyclic}= \gamma_1\bc_{(1,2,3)}+\gamma_2\bc_{(1,2,4)}+\gamma_3\bc_{(1,2,5)}$. If so, testing \eqref{LOF:K-infty} reduces to an $F$--test against the null $H_0:\gamma_1=\gamma_2=\gamma_3=0$.

\begin{table}[h]
\centering
\caption{Estimated power of the $F$ and $R_n^{(1)}$ tests when $K=30$, $m=10$ and normal errors. Critical values and rejection probabilities are computed using $10^5$ simulation runs.}
\label{table:compare:regression}
\begin{tabular}{cccccc} \hline
& True Model & & \multicolumn{3}{c}{Empirical Power}\\
  &  $(\gamma_1,\gamma_2,\gamma_3)$  &  &  & $R_n^{(1)}$  & $F$ \\ \hline
 & $(1/2,1/2,1/2)$ &  &  & 0.366 & 0.709 \\
  & $(3/4,3/4,3/4)$ &  &  & 0.456 & 0.972 \\
  & $(1,1,1)$ &  &  & 0.996 & 0.998 \\
\hline
\end{tabular} 
\end{table}

Table \ref{table:compare:regression} shows that the power increases with $\|\bnu_{\rm cyclic}\|^{2}$. It also indicates that the $F$--test has higher power than $R_n^{(1)}$, although the gap narrows as $\|\bnu_{\rm cyclic}\|^{2}$ and the number of cyclic triads increase. However, when the minimal model is misspecified, the power of the $F$--test can be substantially lower than that of $R_n^{(1)}$. For example, in the setting of Table \ref{table:compare:regression} with $(\gamma_1,\gamma_2,\gamma_3)=(1,1,1)$, the $F$--test based on predictors $\{\bc_{(1,2,6)}, \bc_{(1,2,7)}, \ldots, \bc_{(1,2,15)}\}$ has power $0.348$, well below that of $R_n^{(1)}$. These results show that while the $F$--test is optimal under correct specification, $R_n^{(1)}$ is considerably more robust.

\subsubsection{Assessing $R_n^{(1)}$ on moderately unbalanced graphs}
\label{section:mod-unbalanced-graphs}

We assess the performance of $R_n^{(1)}$ on moderately imbalanced comparison graphs. To do so we let $n_{ij}$ be ${\rm Bin}(m,p)$ RVs with $m\in \{10,15,20,25,30\}$, $p\in\{0.3, 0.4,0.5,1\}$. The true underlying model is assumed to be $\bnu_{\rm cyclic}= \gamma_1\bc_{(1,2,3)}+\gamma_2\bc_{(1,2,4)}+\gamma_3\bc_{(1,2,5)}$. 

\begin{table}[h]
\centering
\caption{Estimated power of $R_n^{(1)}$ test for $K=30$ and normal errors.  Critical values and rejection probabilities are computed using $10^5$ simulation runs.}
\label{table:finite:items}
\begin{tabular}{lllccc} \hline
 & $m$ & $p$ & $(\gamma_1,\gamma_2,\gamma_3)$ &   & Power \\ \hline
 & 10 & 1 & $(1/2,1/2,1/2)$ &   & 0.366 \\
 & 20 & 0.5 & $(1/2,1/2,1/2)$ &   & 0.377 \\ 
 & 30 & 0.33 & $(1/2,1/2,1/2)$ &   & 0.358 \\  \hline
\end{tabular} 
\end{table}

Table \ref{table:finite:items} shows that, with $mp=10$ fixed, the power of $R_n^{(1)}$ remains stable under moderate imbalance and is comparable to that observed for balanced graphs.

\subsubsection{Assessing LOF on highly unbalanced PCGs}
\label{section:high-unbalanced-graphs}

In Sections \ref{section:kend-smith}--\ref{section:mod-unbalanced-graphs} the PCGs studied are either fully balanced or only moderately unbalanced. Consequently, $R_n^{(1)}$ is the appropriate test statistic. Under the regimes described in Theorem \ref{Thm-LOF.2&3} the PCG may be highly unbalanced. This section compares the relative performance of the three test statistics $R_n^{(1)}$, $R_n^{(2)}$, and $R_n^{(3)}$ in such circumstances. Throughout this comparison, we fix  $K = 30$ and set  $\bnu_{\rm cyclic} =\gamma_1 \bc_{(1,2,3)} + \gamma_2 \bc_{(1,2,4)}$. To generate highly unbalanced PCG, we set:
\[
n_{ij}=\begin{cases}
    20 , \quad& \text{if}\quad  j=i+1,\text{ and when }i=1,j=3\\
    10 , \quad& \text{if}\quad  i=2,j=4\\
    \mathrm{Bin}(5, 0.5), \quad& \text{otherwise},
\end{cases}
\]
which yields a deterministic path graph augmented by two additional low--weight deterministic edges and small random independent edges elsewhere. To construct the statistics  $R_n^{(2)}$  and $R_n^{(3)}$, we threshold the comparison counts by retaining only  $n_{ij} \ge 10$ and $n_{ij} \ge 20$, respectively. Their performance is reported on in Table~\ref{table:finite:diff-nu}. 

\begin{table}[h]
\centering
\caption{Estimated powers of tests $R_n^{(1)}$, $R_n^{(2)}$, and $R_n^{(3)}$ for $K=30$, $\bnu_{\rm cyclic}=\,\gamma_1 \bc_{(1,2,3)}+\gamma_2 \bc_{(1,2,4)}$ and normal errors. Critical values and rejection probabilities are computed using $10^5$ simulation runs.}
\label{table:finite:diff-nu}
\begin{tabular}{cccccc} \hline
& True Model & & \multicolumn{3}{c}{Empirical Power}\\
& $(\gamma_1,\gamma_2)$ & & $R_n^{(1)}$ & $R_n^{(2)}$ & $R_n^{(3)}$ \\ \hline
& $(1/2,0)$   & & 0.107 & 0.678 & 0.689 \\
& $(1/2,1/2)$ & & 0.278 & 0.972 & 0.957 \\
\hline
\end{tabular}
\end{table}

The test $R_n^{(1)}$ performs poorly relative to $R_n^{(2)}$ and $R_n^{(3)}$. This is likely because $r_1$, the degrees of freedom of $R_n^{(1)}$, are substantially larger than $r_2$ and $r_3$, the degrees of freedom associated with $R_n^{(2)}$ and $R_n^{(3)}$, respectively, while the triads carrying information on LOF are confined to $\mathcal{E}_2$ and $\mathcal{E}_3$; see Theorem~\ref{theorem:test:powerless}.

\subsection{Large comparison graphs}

We consider complete and random graphs with $K \in \{30,50,70,100,200\}$ and $p_K \in \{1,1/2,1/3\}$. The true model is cyclic, with the first four cyclic triads having unit coefficients:
$$ \bnu_{\rm cyclic}
= \bc_{(1,2,3)} + \bc_{(1,2,4)} + \bc_{(1,2,5)} + \bc_{(1,2,6)}.$$
We assume that $\mathcal{U}=\{1,2,3,4,5,6\}$ is given and $\mathcal{V}\setminus\mathcal{U}=\{7,8,\ldots,K\}$. This corresponds to $20\%$ cyclic items when $K=30$ and to $3\%$ when $K=200$. LOF is tested using \eqref{eq:thm3.2} when $p_K=1$ and \eqref{eq:lof:random} when $p_K<1$. 

There are a variety of ways by which the aforementioned test can be applied. For example, when $p_K = 1$ and the errors are normally distributed, one can estimate $\sqrt{{\rm Var}(R_J)}$ using \eqref{eq:var:RK}, with $J$ instead of $K$ and plugging in an estimate for $\sigma^2$ derived from the observation in $\mathcal{V}\setminus\mathcal{U}$. Alternatively, $\sqrt{{\rm Var}(R_J)}$ can be estimated by bootstrap methods \citep{Davison1997}. When $p_K<1$ Equation \eqref{eq:var:RK} can no longer be applied and a resampling based estimator is necessary. We have experimented with various resampling methods for implementing the tests in \eqref{eq:thm3.2} and \eqref{eq:lof:random} and have found that the method described below works best both in approximating the level of the test and in terms of power. 

Concretely, after observing the PCG $(\mathcal{G},\mathcal{Y})$: $(i)$ estimate $\widehat{\bmu}$ from the full graph; and $(ii)$ calculate $\widehat{\pmb{e}'} =\bY' - \bB' \widehat{\bmu}'$ the residuals from fitting the model on $(\mathcal{G}',\mathcal{Y}')$, the PCG obtained after removing the vertices, edges and samples associated with the items in $\mathcal{U}$. The $b^{th}$ bootstrap PCG (or sample), $b=1,\ldots,B$ is generated by setting
\begin{align*}
Y_{b,ij}^* = (\widehat{\mu}_i-\widehat{\mu}_j) + \hat{e}_{b,ij}', 
\end{align*}
where $\hat{e}_{b,ij}'$ is sampled with replacement from $\widehat{\pmb{e}'}$. For each bootstrap sample we calculate
\begin{align*}
T_b= \binom{J}{2}^{-1}\sum_{i,j\in\mathcal{U}} (Y_{b,ij}^*-(\widehat{\mu }_{i}-\widehat{\mu }_{j}))^{2}.
\end{align*}
Power is assessed by $\sum_{b=1}^B\mathbb{I}(T_b \ge T_{\rm obs})/(B+1)$ where $T_{\rm obs}$ is the observed value of the test statistic. Empirical levels and estimated power are reported in Table~\ref{table:nij=1}. 

\begin{table}[h]
\centering
\caption{Empirical levels and powers of the tests \eqref{eq:thm3.2} and \eqref{eq:lof:random} on Erd\"os--R\'eyni graphs with parameters $(K,p_K)$ and normal errors. Rejection probabilities are computed using $10^5$ simulation runs. The level of the tests is computed when $\bnu_{\rm cyclic}=\pmb{0}$ and the power when $\bnu_{\rm cyclic} = \bc_{(1,2,3)} + \bc_{(1,2,4)}+ \bc_{(1,2,5)} + \bc_{(1,2,6)}$.}
\label{table:nij=1}
\begin{tabular}{c|cccc c| cccc c} \hline
& \multicolumn{5}{c}{Level} & \multicolumn{5}{c}{Power}\\
\diagbox[width=0.8\dimexpr \textwidth/12+2\tabcolsep\relax, height=0.8cm]{$p_K$}{ $K$} & 30 & 50 & 70 & 100 & 200 & 
30 & 50 & 70 & 100 & 200\\ \hline
$1$ & 
0.076 & 0.067 & 0.058 & 0.054 & 0.045 &
0.962 & 0.956 & 0.949  & 0.945  & 0.938\\
$1/2$ & 
0.071 & 0.057 & 0.052 & 0.051 & 0.041 & 
0.545  & 0.553 & 0.515 & 0.524 &0.418\\
$1/3$ & 
0.073 & 0.065 & 0.051 & 0.048 & 0.043 &
0.434 & 0.407 & 0.399 & 0.372 & 0.306\\\hline
\end{tabular}
\end{table}
\FloatBarrier

\noindent The proposed tests are somewhat liberal when the sample size is small, but attain the nominal level of $0.05$ once $K$ exceeds approximately $70$. The tests \eqref{eq:thm3.2}  and \eqref{eq:lof:random} effectively detect LOF even when the graphs are large and sparse. 

\bigskip

\section{Illustrative example} \label{section:illustrative:example}

Sporting competitions are a major source of public interest, study, and entertainment, and they provide a rich setting for statistical analysis and ranking methods. Accurately quantifying relative team strength is important for both competitive and economic reasons, shaping internal decisions and external outcomes such as media coverage, sponsorship, and fan engagement \citep{wang2025nba}. Existing approaches are largely based on linear transitive models \citep{sinuany1988, cassady2005, Langville2012, barrow2013} and relatively little attention has been devoted to LOF testing and the quantification of intransitivity.

To demonstrate the practical relevance of our proposed methodology, we analyze game outcomes from the National Basketball Association (NBA) spanning the 2015--2016 through 2024--2025 seasons, excluding the 2019--2020 season due to disruptions caused by the COVID-19 pandemic. The resulting dataset, sourced from \url{https://www.basketball-reference.com/leagues/}, comprises nine complete seasons, which we treat as temporally consecutive. In each season, there are $K=30$ teams, and the comparison graph is moderately unbalanced. In this setting, it is natural to test for LOF using $R_n^{(1)}$. Table \ref{pval:table} reports the LOF p--value for each season. 

\begin{table}[!ht]
\centering
\caption{Testing LOF: p-values for the seasons for teams violating LOF. Critical values are obtained using $10^5$ Monte Carlo smaples.}
\label{pval:table}
 \resizebox{\textwidth}{!}{  
\begin{tabular}{r rrrr rrrrr}
  \hline
 Season & \rotatebox{0}{2015-16} & \rotatebox{0}{2016-17} & \rotatebox{0}{2017-18} & \rotatebox{0}{2018-19} &
 \rotatebox{0}{2020-21} & \rotatebox{0}{2021-22} & \rotatebox{0}{2022-23} & \rotatebox{0}{2023-24} & \rotatebox{0}{2024-25} \\ 
  \hline
p-value & 0 & 0.00017 & 0 & 0.00008  &
0 & 0 & 0.00001 & 0.04466 & 0 \\ 
   \hline
\end{tabular}}
\end{table}

\medskip

Table \ref{pval:table} shows that the null hypothesis of LST is rejected, with highly significant p-values. Next, we identify teams that contribute to LOF, which we refer to as cyclical teams; see equation~\eqref{bnu:decomposed:linear:cyclic}. A model--based approach for identifying such teams was proposed by \cite{singh2026}. Here, we adopt a more liberal method inspired by \cite{kendall1940}. Let $\hat{\pi}_{ij}$ denote the empirical win fraction of team $i$ over team $j$. A triad $(i,j,k)$ is said to form a cycle if either $\hat{\pi}_{ij} \ge 0.5$, $\hat{\pi}_{jk} \ge 0.5$, and $\hat{\pi}_{ki} \ge 0.5$, or $\hat{\pi}_{ij} \le 0.5$, $\hat{\pi}_{jk} \le 0.5$, and $\hat{\pi}_{ki} \le 0.5$. We iteratively remove cyclical triads via a direct search procedure until no cycles remain. Applying this method to the 2023--2024 season identifies 12 cyclical teams: Chicago Bulls, Denver Nuggets, Orlando Magic, Los Angeles Lakers, Charlotte Hornets, Phoenix Suns, Milwaukee Bucks, Indiana Pacers, Sacramento Kings, Oklahoma City Thunder, Cleveland Cavaliers, and Memphis Grizzlies. This represents a substantial fraction of the league, indicating that intransitivity is not a marginal phenomenon but a pervasive feature of the data.

Table~\ref{table:frequency:NBA} reports the number of seasons in which each team is identified as cyclical.

\begin{table}[!ht]
\centering
\caption{Teams violating LOF and their frequency for seasons $\{$2020-2021, \dots, 2024-2025$\}$.}
\label{table:frequency:NBA}
\setlength\tabcolsep{2pt}
\begin{tabular}{r|llllllllllllllllllllllllllllll}
  \hline
{Team} & \rotatebox{90}{Chicago Bulls} & \rotatebox{90}{Denver Nuggets} & \rotatebox{90}{Atlanta Hawks} & \rotatebox{90}{Brooklyn Nets} & \rotatebox{90}{Cleveland Cavaliers} & \rotatebox{90}{Los Angeles Clippers} & \rotatebox{90}{Indiana Pacers} & \rotatebox{90}{Miami Heat} & \rotatebox{90}{Phoenix Suns} & \rotatebox{90}{Sacramento Kings} & \rotatebox{90}{Charlotte Hornets} & \rotatebox{90}{Memphis Grizzlies} & \rotatebox{90}{Minnesota Timberwolves} & \rotatebox{90}{Oklahoma City Thunder} & \rotatebox{90}{Toronto Raptors} & \rotatebox{90}{Boston Celtics} & \rotatebox{90}{Dallas Mavericks} & \rotatebox{90}{Los Angeles Lakers} & \rotatebox{90}{New York Knicks} & \rotatebox{90}{Orlando Magic} & \rotatebox{90}{Utah Jazz} & \rotatebox{90}{Washington Wizards} & \rotatebox{90}{Detroit Pistons} & \rotatebox{90}{Golden State Warriors} & \rotatebox{90}{Houston Rockets} & \rotatebox{90}{Milwaukee Bucks} & \rotatebox{90}{New Orleans Pelicans} & \rotatebox{90}{San Antonio Spurs} & \rotatebox{90}{Philadelphia 76ers} & \rotatebox{90}{Portland Trail Blazers} \\ \hline
  Frequency & 8 & 8 & 6 & 6 & 6 & 6 & 5 & 5 & 5 & 5 & 4 & 4 & 4 & 4 & 4 & 3 & 3 & 3 & 3 & 3 & 3 & 3 & 2 & 2 & 2 & 2 & 2 & 2 & 1 & 1 \\ 
   \hline
\end{tabular}
\end{table}

To quantify the structural overlap of cyclic teams across seasons, we compute the pairwise Jaccard similarity index \citep{bag2019} between the sets of cyclic teams in any two seasons $t$ and $t'$, defined as 
$$J(t,t') = \frac{|\mathcal{S}_t \cap \mathcal{S}_{t'}|}{|\mathcal{S}_t \cup \mathcal{S}_{t'}|},$$ 
where $\mathcal{S}_t \subset \mathcal{V}$ denotes the set of cyclic teams in season $t$. This yields a symmetric $9 \times 9$ similarity matrix with unit diagonal. Values approaching $1$ indicate highly conserved cyclic structure across seasons, whereas values near $0$ indicate largely disjoint, and hence more stochastic, violations of LST. The Jaccard indices for the NBA seasons are reported in Table~\ref{Jaccard:table}.

\begin{table}[!ht]
\centering
\caption{The Jaccard indices for the seasons for cyclic teams.}
\label{Jaccard:table}
 \resizebox{\textwidth}{!}{ 
\begin{tabular}{r|rrrrrrrrr}
  \hline
 & 2015-16 & 2016-17 & 2017-18 & 2018-19 & 2020-21 & 2021-22 & 2022-23 & 2023-24 & 2024-25 \\ 
  \hline
2015-16 & 1.00 & 0.62 & 0.35 & 0.22 & 0.26 & 0.16 & 0.28 & 0.31 & 0.29 \\ 
  2016-17 & 0.62 & 1.00 & 0.30 & 0.25 & 0.29 & 0.14 & 0.30 & 0.50 & 0.39 \\ 
  2017-18 & 0.35 & 0.30 & 1.00 & 0.17 & 0.45 & 0.42 & 0.33 & 0.24 & 0.29 \\ 
  2018-19 & 0.22 & 0.25 & 0.17 & 1.00 & 0.27 & 0.37 & 0.29 & 0.25 & 0.13 \\ 
  2020-21 & 0.26 & 0.29 & 0.45 & 0.27 & 1.00 & 0.40 & 0.38 & 0.29 & 0.27 \\ 
  2021-22 & 0.16 & 0.14 & 0.42 & 0.37 & 0.40 & 1.00 & 0.35 & 0.14 & 0.24 \\ 
  2022-23 & 0.28 & 0.30 & 0.33 & 0.29 & 0.38 & 0.35 & 1.00 & 0.24 & 0.50 \\ 
  2023-24 & 0.31 & 0.50 & 0.24 & 0.25 & 0.29 & 0.14 & 0.24 & 1.00 & 0.47 \\ 
  2024-25 & 0.29 & 0.39 & 0.29 & 0.13 & 0.27 & 0.24 & 0.50 & 0.47 & 1.00 \\ 
   \hline
\end{tabular}}
\end{table}

To formalize the season--over--season persistence of LOF violations, we model each team's LOF status as a discrete--time Markov chain with binary state space $\Omega = \{0,1\}$. Let $X_{i,t} = \mathbb{I}(i \in \mathcal{S}_t)$ denote the state of team $i$ in season $t$, where state $1$ indicates membership in the set of teams violating LOF and state $0$ indicates compliance with LST. Assuming time--homogeneity, the first-order transition probabilities $p_{kl} = \mathbb{P}(X_{i,t+1} = l \mid X_{i,t} = k)$, for $k,l \in \{0,1\}$, are estimated by maximum likelihood from pooled adjacent-season transitions across all teams, yielding
\[
\hat{p}_{kl} = \frac{\sum_{i \in \mathcal{V}} \sum_{t} \mathbb{I}(X_{i,t} = k, X_{i,t+1} = l)}{\sum_{i \in \mathcal{V}} \sum_{t} \mathbb{I}(X_{i,t} = k)}.
\]
Applying this procedure to the NBA data yields the following estimated transition matrix:
\begin{center}
\begin{tabular}{r|rr}
  \hline
 & 0 & 1 \\ 
  \hline
0 & 0.61 & 0.39 \\ 
1 & 0.49 & 0.51 \\ 
  \hline
\end{tabular}
\end{center}
These estimates indicate that a cyclic team has approximately a 50\% probability of remaining cyclic in the subsequent season, suggesting moderate persistence of LOF violations over time.

Finally, we investigate whether cyclic teams identified in season $t$ can predict the presence of LST violations in complete comparison graphs sampled in season $t+1$. To this end, following each season we partition the teams into two disjoint sets: cyclic and transitive. Applying this procedure across all available seasons, we treat the identified cyclic teams as candidates (our $\mathcal{U}$) and evaluate the power of the test \eqref{eq:thm3.2} using $10^5$ bootstrap samples of complete comparison graphs constructed from the subsequent season’s data. The resulting empirical rejection probabilities, evaluated at the $5\%$ and $10\%$ significance levels, are reported in Table~\ref{predict:table}.

\begin{table}[!ht]
\centering
\caption{Rejection probabilities of \eqref{eq:thm3.2} for detecting LOF. }
\label{predict:table}
\begin{tabular}{l|llllllll} \hline
Current   & 2015-16  & 2016-17  & 2017-18  & 2018-19  & 2020-21  & 2021-22  & 2022-23  & 2023-24  \\
Future & 2016-17  & 2017-18  & 2018-19  & 2020-21  & 2021-22  & 2022-23  & 2023-24  & 2024-25  \\ \hline
Level--$0.05$   & 0.317 & 0.614 & 0.022 & 0.253 & 0.699 & 0.152 & 0.188 & 0.165 \\
Level--$0.10$ & 0.470 & 0.268 & 0.048 & 0.352 & 0.769 & 0.236 & 0.282 & 0.236 \\ \hline
\end{tabular}
\end{table}

These results, together with the highly significant p--values reported above, suggest that the set of teams violating LOF is not stable across seasons. Consistent with this, the Jaccard similarity matrix and the transition matrix indicate substantial turnover in the identity of cyclic teams. Consequently, knowledge of cyclic teams in the current season does not reliably predict LOF in the subsequent season, indicating that while lack of fit is pervasive, it does not exhibit strong temporal stability or predictive structure.

\section{Discussion} \label{section:discussion}

Linear stochastic transitivity, as exemplified by \eqref{nu.ij=mu.i-mu.j}, underpins much of the analysis of paired comparison data. However, empirical violations of this assumption are frequently reported across various applied domains. For example, in competitive sports, cyclic dominance naturally arises but it is not captured by standard ranking models \citep{Cain2000, vanOurs2024, vanOurs2025, singh2026}. Despite these well--documented violations, the methodological framework for formally detecting LOF has remained underdeveloped, and the test introduced by \cite{kendall1940} has been the predominant approach. As demonstrated in Section \ref{sec.num.study}, this test exhibits poor power for detecting LOF when $K$ is finite and is fundamentally unsuited for large sparse comparison graphs.

To address these limitations, this paper proposes a collection of tests for LOF. Each of the proposed tests is adapted to a different regime that governs the growth of the number of paired comparisons. For example, when $K$ is finite and Condition \ref{Condition:all.nij} holds, then using the statistic $R_{n}^{(1)}$ is appropriate. The statistical power of this test is comparable, in some settings, to that of the benchmark $F$-test, which relies on prior knowledge of the cyclic parameters. If the regime is best described by Conditions \ref{Condition:sqrt(m)} or \ref{Condition:sqrt(n)} then the statistics $R_{n}^{(2)}$ or $R_{n}^{(3)}$ should be used. Our simulations show that the latter are considerably more powerful in highly unbalanced graphs. Theorem~\ref{theorem:test:powerless} establishes the conditions under which the tests $R_{n}^{(1)},~ R_{n}^{(2)}$ and $R_{n}^{(3)}$ are consistent. In addition, the paper also addresses the situation where the number of items satisfies $K\to \infty$. Within this framework, several scenarios are considered: first, the balanced case $n_{ij}=m>1$ is analyzed; then we set $n_{ij}=1$, and finally we allow sparse graphs where $n_{ij}\in\{0,1\}$. The limiting distribution of the proposed statistics is derived, and conditions guaranteeing consistency are discussed, cf. Proposition \ref{prop 3.1}. Also note that when all $n_{ij}\le 1 $ additional information is required, either in the form of the localization of the set of cyclic items or knowledge of $\sigma^2$.   

As noted in the Introduction, the current paper provides a road map for developing tests for LOF in other paired comparison models, the best known of which is the Bradley--Terry (BT) model. In fact, the framework proposed here can be directly applied to the BT model. For example, the BT analogue of the statistic \eqref{Eq.Rn} is 
\begin{align*}
 \sum_{(i,j)\in \mathcal{E}_1}n_{ij}(\frac{{S}_{ij}}{n_{ij}}-\frac{\exp(\widehat{\mu }_{i}-\widehat{\mu }_{j})}{1+\exp(\widehat{\mu }_{i}-\widehat{\mu }_{j})})^{2},   
\end{align*}
where $S_{ij} = \sum_{k=1}^{n_{ij}}Y_{ijk}$ and $Y_{ijk}$ are IID Bernoulli RVs with success probability $\exp(\mu_i-\mu_j)/(1+\exp(\mu_i-\mu_j))$. Here $\widehat{\mu }_{1},\ldots, \widehat{\mu }_{K}$ are the merit parameters estimated from the BT model. The test above is motivated by \cite{GSK1969} and \cite{beaver1977} who showed that certain generalized linear models can be analyzed as linear models. Recently, \cite{hendrickx2020} emphasized the robustness of this formulation in high-dimensional settings; see also \cite{oliveira2018} for a related approach. Counterparts of $R_n^{(2)}$ and $R_n^{(3)}$ are constructed similarly. Preliminary numerical experiments indicate that these tests are effective in practice. Analogues of the tests proposed for large sparse graphs described in Section \ref{section:infinite:K} are also possible.  Nevertheless, the resulting procedures do not account for the binary nature of the observations and are therefore not fully satisfactory. 
Thus, the development of more powerful, principled tests for the BT and other binary PCD models that leverage the ideas introduced in this paper remains an open and important problem for future work. 

A limitation of the current approach to large sparse PCGs is the assumption that the set of potentially cyclic items, $\mathcal{U}$, is known a priori, whereas in practice it is typically unknown. Identifying the subset $\mathcal{U}$ is itself a challenging problem and may be as difficult as detecting lack of fit in sparse regimes. While screening procedures based on normal approximations are possible, numerical experiments indicate that they are inherently low--powered, as they rely on single--pair deviations rather than aggregated cyclic structure. Developing methods for reliably identifying $\mathcal{U}$ in sparse, large PCGs, or characterizing the limits of such identification, is therefore an important direction for future work.

Condition \ref{condition:U,V} stipulates that $J = o(K)$, i.e., cyclicality is localized to a small fraction of items. Although this is a reasonable assumption in many applications, a different type of result may be required when $J = O(K)$. However, the family of statistics proposed here does not appear to be well suited for this setting. An alternative family based on statistics with kernels of the form 
\begin{align*}
\| \widehat{\bmu}(\mathcal{V}\setminus\mathcal{U},\mathcal{V}) - \widehat{\bmu}(\mathcal{V}\setminus\mathcal{U},\mathcal{V}\setminus\mathcal{U}) \|_{(\mathcal{G},\mathcal{Y})}^2 
\end{align*}
may be appropriate. Here, $\widehat{\bmu}(\mathcal{A},\mathcal{B})$ denotes a vector (or subvector) of merits indexed by $i\in \mathcal{A}$ and estimated using the PCG on $\mathcal{B}$ and $\|\cdot\|_{(\mathcal{G},\mathcal{Y})}$ denotes a possibly data dependent norm. The statistic above compares the merits of items in $\mathcal{V}\setminus\mathcal{U}$ estimated in two ways: once using the full graph and once using only the acyclic subgraph. Since this test statistic does not belong to the family of tests proposed here, we do not pursue it further.

Finally, once LOF is detected, we conclude that $\bnu = \bnu_{\rm linear} + \bnu_{\rm cyclic}$ with $\bnu_{\rm cyclic} \neq \pmb{0}$. In this case, $\bnu$ is likely to be intransitive, but not necessarily so. As shown in Example~2.2 of \cite{singh2026}, intransitivity occurs when the cyclic component $\bnu_{\rm cyclic}$ dominates the linear component $\bnu_{\rm linear}$. \cite{singh2026} shows how $\bnu_{\rm cyclic}$ can be decomposed, modeled, and estimated from the data as a minimal, interpretable linear combination of the cyclical vectors $\{\bc_{(i,j,k)} \in \mathbb{R}^{\binom{K}{2}} : 1 \leq i < j < k \leq K\}$ described in Section \ref{section:infinite:K}. Such models effectively capture  cyclical preferences in the data.

\section*{Acknowledgments}

The work of Rahul Singh was supported by New Faculty Seed Grant No. MI03038G the Indian Institute of Technology Delhi, India. The work of Ori Davidov was partially supported by the Israeli Science Foundation Grant No. 2200/22 and gratefully acknowledged.

\bibliographystyle{apalike}
\bibliography{glm_ref}

\appendix

\section{Proofs}

For $s\in\{1,2,3\}$ let 
$$\bQ_s= \bD_s^{1/2} (\pmb{I}-\bB (\bB^{\top} \bD_1\bB)^{+} \bB^{\top} \bD_1)\bD_1^{-1/2},$$ 
where the matrices $\bD_s$'s are defined in the text and 
so $\pmb{\Omega}_1=\sigma^2 \bQ_1\bQ_1^{\top}$ and $\pmb{\Omega}_s=\sigma^2 \lim\, \bQ_s\bQ_s^{\top}$ for $s=2,3$.

\begin{lemma} \label{lemma.M.converge}
Let $\bP_{1}$ be the $|\bnu|\times |\bnu|$ diagonal matrix whose $(i,j)^{th}$ diagonal element is $1$ if $(i,j)\in\mathcal{E}_1$ and $0$ otherwise, and $\bX=\bD_1^{1/2}\bB$ and $\bX_{n}$ is the submatrix of $\bX$ corresponding to non-zero rows in $\bP_{1}$. Then, $\bP_{\bX_{n}} = \bX_{n}(\bX_{n}^\top \bX_{n})^{+} \bX_{n}^\top$ is the orthogonal projector onto column space of $\bX_{n}$. Furthermore,
\[
  \bQ_1=\, \bP_1(\pmb{I} - \bX_{n} (\bX_{n}^\top \bX_{n})^{+} \bX_{n}^\top).
  \]
Consequently, $\bQ_1$ is (the restriction to $\bP_{1}$ of) an orthogonal projector complement, and hence its entries are uniformly bounded (in fact, in $[-1,1]$). Furthermore,
\[
  {\rm rank}\,(\bQ_1\bQ_1^{\top})= {\rm rank}\,(\bQ_1) = |\mathcal{E}_1| - (K-1).
  \]
\end{lemma}

\begin{proof}
Notice that 
\[
  \bD_1^{1/2}(\bD_1^{+})^{1/2} = \bP_1 \quad \text{and} \quad
  \bX^\top \bX=\bB^\top \bD_1 \bB,
  \]
consequently, 
\begin{align*}
\bQ_1
&= \bD_1^{1/2}(\pmb{I} - \bB(\bB^\top \bD_1 \bB)^{+} \bB^\top \bD_1)\bD_1^{+})^{1/2}\\
&= \bP_1 - \bD_1^{1/2}\bB(\bB^\top \bD_1 \bB)^{+} \bB^\top \bD_1 (\bD_1^{+})^{1/2}.
\end{align*}
Now $\bD_1(\bD_1^{+})^{1/2} = \bD_1^{1/2}\bP_1$, so that $\bB^\top \bD_1(\bD_1^{+})^{1/2} = \bB^\top \bD_1^{1/2}\bP_1 = \bX^\top \bP_1$. Since $\bX$ has zero rows outside $\mathcal{E}_1$, this reduces to $\bX^\top \bP_1 = \bX_{n}^\top$. Therefore
\[
  \bD_1^{1/2}\bB(\bB^\top \bD_1 \bB)^{+} \bB^\top \bD_1(\bD_1^{+})^{1/2} 
  = \bX(\bX^\top \bX)^{+} \bX^\top \bP_1 
  = \bP_1 \bX_{n}(\bX_{n}^\top \bX_{n})^{+} \bX_{n}^\top.
  \]
Therefore, we obtain
\[
  \bQ_1 \;=\; \bP_1(\pmb{I} - \bX_{n} (\bX_{n}^\top \bX_{n})^{+} \bX_{n}^\top).
  \]
This proves the first part. Next, since $\pmb{I}-\bP_{\bX_{n}}$ is an orthogonal projector, its entries are bounded by $1$ in absolute value, and multiplication by $\bP_1$ simply pads zero rows/columns. Hence all entries are finite and uniformly bounded, independent of scaling of the entries of $\bD_1$. Now,
\[
  {\rm rank}\,(\bQ_1) = |\mathcal{E}_1| - {\rm rank}(\bX_{n}) = |\mathcal{E}_1| - {\rm rank}(\bB^T \bD_1 \bB)= |\mathcal{E}_1| - (K-1),
  \]
as for a connected graph ${\rm rank}(\bB^T \bD_1 \bB)=K-1$. Therefore,
\[
  {\rm rank}\,(\bQ_1\bQ_1^{\top})= {\rm rank}\,(\bQ_1) = |\mathcal{E}_1| - (K-1).
  \]
\end{proof}

\subsubsection*{Proof of Theorem \ref{Thm-LOF.1}:}
\begin{proof}
Let $\pmb{U}_{n}$ be the random vector of dimension $|\bnu|$ whose $(i,j)^{th}$ components is $U_{ij}=\sqrt{n_{ij}}(S_{ij}/n_{ij}-(\widehat{\mu }_{i}-\widehat{\mu }_{j}))$ when $(i,j)\in\mathcal{E}_1$ and $0$ otherwise. Note that 
$$
  \widehat\mu_i -\widehat\mu_j = (\pmb e_i- \pmb e_j)^{\top}\widehat\bmu = (\pmb e_i- \pmb e_j)^{\top} \bN^{+}\bS.
$$
  Let $\overline{\bS}_{\rm ALL}$ be the $|\bnu|\time 1$ vector whose $(i,j)^{th}$ element is equal to $S_{ij}/n_{ij}$'s if $(i,j)\in\mathcal{E}_1$ and $0$ otherwise. Now, using some algebra it can be shown that
\begin{align*} \label{thm1.pf.eq1}
\pmb{U}_n =&~ \bD_1^{1/2} (\pmb{I}-\bB (\bB^{\top} \bD_1\bB)^{+} \bB^{\top} \bD_1) \overline{\bS}_{\rm ALL}\\
=&~ \bD_1^{1/2} (\pmb{I}-\bB (\bB^{\top} \bD_1\bB)^{+} \bB^{\top} \bD_1)(\bD_1^{+})^{1/2} \bD_1^{1/2}\overline{\bS}_{\rm ALL}\\
=&~ \bQ_1\bD_1^{1/2}\overline{\bS}_{\rm ALL}, \numberthis
\end{align*}
where $\bQ_1= \bD_1^{1/2} (\pmb{I}-\bB (\bB^{\top} \bD_1\bB)^{+} \bB^{\top} \bD_1)(\bD_1^{+})^{1/2}$. Clearly, $\bD_1^{1/2}(\overline{\bS}_{\rm ALL}-\bnu) \stackrel{d}{=} \mathcal{N} (\bzero, \sigma^2 \pmb{H})$, where $\pmb{H}$ is a $|\bnu|\times |\bnu|$ diagonal matrix whose diagonal elements are $1$ if $(i,j)\in\mathcal{E}_1$ and $0$ otherwise. It follows from relation \eqref{thm1.pf.eq1}, that 
$$
\pmb{U}_n \stackrel{d}{=} \mathcal{N} (\bzero, \pmb{\Omega}_1) \text{ where } \pmb{\Omega}_1=\sigma^2 \bQ_1\pmb{H}\bQ_1^\top= \sigma^2 \bQ_1\bQ_1^\top.
$$
Next observe that
\begin{equation*}
R_{n}^{(1)}=\sum_{(i,j)\in \mathcal{E}}n_{ij}(\overline{S}_{ij}-(\widehat{\mu }_{i}-\widehat{\mu }_{j}))^{2}=\pmb{U}_{n}^\top\pmb{U}_{n}.
\end{equation*}
Further note that since $\pmb{\Omega}_1$ is symmetric and nonnegative definite we may write its spectral decomposition as $\pmb{\Omega}_1=\pmb{O}^\top\pmb{\Lambda }\pmb{O}$ where $\pmb{O}$ is an orthogonal matrix whose columns are the eigenvectors of $\pmb{\Omega}_1$ and $\pmb{\Lambda }$ is a diagonal matrix with nonnegative elements which are the eigenvalues of $\pmb{\Omega}_1$. Clearly $\pmb{U}_n=\pmb{\Omega}_1^{1/2}\bZ$ where $\bZ$ is a $\mathcal{N}_{|\bnu|}(0,\pmb{I})$ RV. It follows that
\begin{equation*}
\pmb{U}_n^\top\pmb{U}_n=(\pmb{\Omega}_1^{1/2}\bZ
)^\top(\pmb{\Omega}_1^{1/2}\bZ)=\bZ^\top\pmb{
  \Psi}_n \bZ=\bZ^\top\pmb{O}^\top\pmb{\Lambda EZ}=(
    \pmb{O}\pmb{Z})^\top\pmb{\Lambda }(\pmb{O}\pmb{Z})\stackrel{d}{=}\pmb{Z}^\top\pmb{\Lambda }\pmb{Z}= \sum_{i=1}^{|\bnu|}\lambda _{i}Z_{i}^{2},
\end{equation*}
since $\pmb{O}\pmb{Z}$ and $\bZ$  have the same distribution. Further, it is straightforward to verify that
\begin{align*}
\pmb{\Omega}_1 = \sigma^2 (\bD_1^{+})^{1/2} (\pmb{I}-\bB (\bB^{\top} \bD_1\bB)^{+} \bB^{\top} \bD_1) \, \bD_1^{+} (\pmb{I}-\bB (\bB^{\top} \bD_1\bB)^{+} \bB^{\top} \bD_1)^{\top} (\bD_1^{+})^{1/2}.
\end{align*}
Finally, the rank identity follows from Lemma \ref{lemma.M.converge}. This proves the first part.

Next, if $\bnu =\pmb{l}+\pmb{\delta}$ where $\pmb{l}\in\mathcal{L}$ and $\pmb{\delta}\in\mathcal{C}$ then
  \begin{align*}
  U_{ij}=\sqrt{n_{ij}}(\overline{S}_{ij}-(\widehat{\mu }_{i}-\widehat{\mu }_{j})) 
  =\sqrt{n_{ij}}(\overline{S}_{ij}-(\widehat{\mu }_{i}-\widehat{\mu }_{j}-{\delta}_{ij})+ {\delta}_{ij})= U_{ij}^* + \sqrt{{n_{ij}}}\,{\delta}_{ij},
  \end{align*}
  where $ U_{ij}^*= \sqrt{n_{ij}}(\overline{S}_{ij}-(\widehat{\mu }_{i}-\widehat{\mu }_{j}- {\delta}_{ij}))$. Notice that, here $U_{ij}^*$'s behave same as the $U_{ij}$ when $\bnu\in\mathcal{L}$. Consequently $\pmb{U}_n\sim \mathcal{N}_{|\bnu|} (\bD^{1/2}\,\pmb{\delta},\, \pmb{\Omega}_1)$. It follows that 
\begin{align*}
R_n^{(1)}=&~ \pmb{U}_n^{\top} \pmb{U}_n= 
  ((\pmb{\Omega}_1^{+})^{1/2}\pmb{U}_n )^\top\pmb{\Omega}_1 ((\pmb{\Omega}_1^{+} )^{1/2}\pmb{U}_n)= ((\pmb{\Omega}_1^{+})^{1/2} \pmb{U}_n
  )^\top\pmb{O}_1^{\top}\pmb{\Lambda}\pmb{O}_1 ((\pmb{\Omega}_1^{+})^{1/2}\pmb{U}_n)\\
=&~ \pmb{V}^{\top} \pmb{\Lambda} \pmb{V},
\end{align*}
where $\pmb{V}= \pmb{O}_1 (\pmb{\Omega}_1^{+})^{1/2}\pmb{U}$ is a $\mathcal{N}(\pmb{\phi}_1, \pmb{H})$ RV, $\pmb{\phi}_1= \pmb{O}_1 (\pmb{\Omega}_1^{+})^{1/2}\bD_1^{1/2}\, \pmb{\delta}$ and $\pmb{H}$ is a $|\bnu|\times |\bnu|$ diagonal matrix whose $(i,j)^{th}$ diagonal elements is $1$ if $(i,j)\in\mathcal{E}_1$ and $0$ otherwise. It immediately follows that 
$$
R_n^{(1)}= \sum_{i=1}^{r_1} \lambda_i (Z_i+ \phi_i)^2,
$$
where $\phi_i$ is the element of $\pmb{\phi}_1$ corresponding with the $i^{th}$ nonzero diagonal element of $\pmb{\Lambda}$ denoted by $\lambda_i$. 
\end{proof}

\subsubsection*{Proof of Theorem \ref{Thm-LOF.2&3}:}
\begin{proof}
Following steps similar to those in the proof of Theorem \ref{Thm-LOF.1}, we have
$$R_{n}^{(2)}=\pmb{U}_{n}^{{(2)}\top} \pmb{U}_{n}^{(2)},$$
  where
\begin{align*}
\pmb{U}_{n}^{(2)} 
=&~\bD_2^{1/2} (\pmb{I} - \bB (\bB^{\top} \bD_1\bB)^{+} \bB^{\top} \bD_1)\;\overline{\bS}_{\rm ALL}\\
=&~ \bD_2^{1/2} (\pmb{I} \bB (\bB^{\top} \bD_1\bB)^{+} \bB^{\top} \bD_1)(\bD_1^{+})^{1/2} \bD_1^{1/2}\;\overline{\bS}_{\rm ALL}\\
=&~ \pmb{Q}_2 \bD_1^{1/2}\;\overline{\bS}_{\rm ALL},
\end{align*}
where $\pmb{Q}_2=\bD_2^{1/2} (\pmb{I}- \bB (\bB^{\top} \bD_1\bB)^{+} \bB^{\top} \bD_1)(\bD_1^{+})^{1/2}$. Following arguments similar to those in Lemma \ref{lemma.M.converge}, it follows that for all large $n$, $\pmb{Q}_2$ has uniformly bounded entries. Furthermore, by the CLT we have $\bD_1^{1/2}(\overline{\bS}_{\rm ALL}-\bnu) \Rightarrow \mathcal{N} (\bzero, \sigma^2 \pmb{H})$. Now if $\lim_n \pmb{Q}_2$ exists, then it follows by Slutzky's Theorem that $\pmb{U}_{n}^{(2)}\Rightarrow \pmb{U}$ where $\pmb{U}$ follows a $\mathcal{N}_{|\bnu|}(0,\; \pmb{\Omega}_2)$, where $\pmb{\Omega}_2=\lim_n \pmb{Q}_2\;\pmb{Q}_2^{\top}$. 

Next, $R_{n}^{(2)}=\pmb{U}_{n}^{{(2)}\top} \pmb{U}_{n}^{(2)}$, and 
since $\pmb{\Omega}_2$ is symmetric and nonnegative definite we may write its spectral decomposition as $\pmb{\Omega }_2=\pmb{O}^\top\pmb{\Lambda }\pmb{O}$ where $\pmb{O}$ is an orthogonal matrix whose columns are the eigenvectors of $\pmb{\Omega }_2$ and $\pmb{\Lambda }$ is a diagonal matrix with nonnegative elements which are the eigenvalues of $\pmb{\Omega}_2$. Clearly $\pmb{U}=\pmb{\Omega}_2^{1/2}\bZ$ where $\bZ$ is a $\mathcal{N}_{|\bnu|}(0,\pmb{I})$ RV. It follows that
\begin{equation*}
\pmb{U}^\top\pmb{U}=(\pmb{\Omega}_2^{1/2}\bZ
)^\top(\pmb{\Omega}_2^{1/2}\bZ)=\bZ^\top\pmb{
  \Omega}_2 \bZ=\bZ^\top\pmb{O}^\top\pmb{\Lambda EZ}=(
    \pmb{O}\pmb{Z})^\top\pmb{\Lambda }(\pmb{O}\pmb{Z})\stackrel{d}{=}\pmb{Z}^\top\pmb{\Lambda }\pmb{Z}= \sum_{i=1}^{|\bnu|}\lambda _{i}Z_{i}^{2},
\end{equation*}
since $\pmb{O}\pmb{Z}$ and $\bZ$  have the same distribution. Now by Lemma \ref{lemma.M.converge} $r_2={\rm rank}(\bM_2)= |\mathcal{E}_2|-(K-1)$ so only $r_2$ eigenvalues of $\pmb{\Psi}_2$ are positive. We conclude that $R_{n}^{(2)}\Rightarrow\sum_{i=1}^{r_2}\lambda _{i}Z_{i}^{2}$ as stated. For $R_{n}^{(3)}$, the asymptotic distribution can be derived by following the same steps as in the case of $R_{n}^{(2)}$, so we skip the details. 

Next, we have to obtain the distribution of $R_n^{(s)}$, for $s\in\{2,3\}$, under the local alternatives. We obtain distribution for $R_n^{(3)}$ only, similar steps follows for $R_n^{(2)}$, and hence skipped. Assume notations in the proof of Theorem \ref{Thm-LOF.2&3}. Under the stated conditions
  \begin{align*}
  U_{ij}^{(3)}=\sqrt{n_{ij}}(\overline{S}_{ij}-(\widehat{\mu }_{i}-\widehat{\mu }_{j})) 
  =\sqrt{n_{ij}}(\overline{S}_{ij}-(\widehat{\mu }_{i}-\widehat{\mu }_{j}-n^{-1/2}{\delta}_{ij})+ n^{-1/2}{\delta}_{ij})= U_{ij}^* + \sqrt{\frac{n_{ij}}{n}}\,{\delta}_{ij},
  \end{align*}
  where $ U_{ij}^*= \sqrt{n_{ij}}(\overline{S}_{ij}-(\widehat{\mu }_{i}-\widehat{\mu }_{j}-n^{-1/2}{\delta}_{ij}))$. Consequently $\pmb{U}_n^{(3)} \Rightarrow \pmb{U}$ where $\pmb{U}$ is a $\mathcal{N}_{|\bnu|} (\pmb{\Xi}_3^{1/2}\,\pmb{\delta}_3, \pmb{\Omega}_3)$ RV. It follows that 
  \begin{align*}
  R_n^{(3)}=&~ \pmb{U}_n^{(3){\top}} \pmb{U}_n^{(3)} \Rightarrow 
  \pmb{U}^{\top} \pmb{U}= 
    ((\pmb{\Omega}_3^{+})^{1/2}\pmb{U})^\top \pmb{\Omega}_3((\pmb{\Omega}_3^{+})^{1/2}\pmb{U})= ((\pmb{\Omega}_3^{+})^{1/2}\pmb{U}
    )^\top\pmb{O}_3^{\top}\pmb{\Lambda}\pmb{O}_3 ((\pmb{\Omega }_3^{+})^{1/2}\pmb{U})\\
  =&~ \pmb{V}^{\top} \pmb{\Lambda} \pmb{V},
  \end{align*}
  where $\pmb{V}= \pmb{O}_3 (\pmb{\Omega}_3^{+})^{1/2}\pmb{U}$ is a $\mathcal{N}(\pmb{\phi}_3, \pmb{H})$ RV, $\pmb{\phi}_3= \pmb{O}_3 (\pmb{\Omega }_3^{+})^{1/2}\pmb{\Xi}_3^{1/2}\, \pmb{\delta}$ and $\pmb{H}$ is a $|\bnu|\times |\bnu|$ diagonal matrix whose $(i,j)^{th}$ diagonal elements is $1$ if $(i,j)\in\mathcal{E}_3$ and $0$ otherwise. It immediately follows that 
  $$R_n^{(3)}= \sum_{i=1}^{r_3} \lambda_i (Z_i+ \phi_i)^2,$$
where $\phi_i$ is the element of $\pmb{\phi}_3$ corresponding with the $i^{th}$ nonzero diagonal element of $\pmb{\Lambda}$ denoted by $\lambda_i$. 
\end{proof}

\subsubsection*{Proof of Theorem \ref{theorem:test:powerless}:}
\begin{proof}
From Theorem \ref{Thm-LOF.1}, we have $\pmb{\phi}_1= \pmb{O}_1 (\pmb{\Omega}_1^{+})^{1/2}\bD_1^{1/2}\, \pmb{\delta}$. Therefore, if for all triads $(i_1,i_2,i_3)$'s satisfying  \eqref{eq:linear:nu:restriction}, $(i_1,i_2,i_3)\notin \mathcal{E}_1$ then $\bD_1^{1/2}\, \pmb{\delta}=\pmb{0}$, implying that $\pmb{\phi}_1=\pmb{0}$. Similarly, from \ref{Thm-LOF.2&3}, we have $\pmb{\phi}_s= \pmb{O}_s (\pmb{\Omega }_s^{+})^{1/2}\pmb{\Xi}_s^{1/2}\, \pmb{\delta}$ for $s=2,3$. Therefore, if for all triads $(i_1,i_2,i_3)$'s satisfying  \eqref{eq:linear:nu:restriction}, $(i_1,i_2,i_3)\notin \mathcal{E}_s$ then $\pmb{\Xi}_s^{1/2}\, \pmb{\delta}=\pmb{0}$, implying that $\pmb{\phi}_s=\pmb{0}$.  

Consequently, if for all triads $(i_1,i_2,i_3)$'s satisfying  \eqref{eq:linear:nu:restriction}, $(i_1,i_2,i_3)\notin \mathcal{E}_s$ then the distribution of the test statistic are same under the null and the local alternatives. Therefore, the power of the tests  $R_n^{(s)}, s\in\{1,2,3\}$ does not exceed $\alpha$ and, therefore, the tests are not consistent.  
\end{proof}

\subsubsection*{Proof of Theorem \ref{Thm-clt-RKm}:}

The proof uses some properties of the matrix $\bB$ summarized in the following remark.
\begin{remark} \label{remark:gof:1}
The incidence matrix $\bB$ of a complete graph with $K$ vertices and $\frac{K(K-1)}{2}$ edges is a $\frac{K(K-1)}{2} \times K$ matrix. Each row corresponds to an edge and contains exactly one $+1$ and one $-1$ indicating the two connected vertices. 
Since $\bB$ represents a complete graph, the product $\bB\bB^T \in \mathbb{R}^{\frac{K(K-1)}{2} \times \frac{K(K-1)}{2}}$. Further, the rank of $\bB\bB^T$ is $k-1$, as the complete graph has rank deficiency of 1 in its Laplacian. Moreover, the eigenvalues of $\bB\bB^T$ are: $K$ with multiplicity $K-1$ and $0$ with multiplicity $\frac{K(K-1)}{2} - (K-1)$. 
Using spectral decomposition: $\bB\bB^T = \bQ \pmb{\Lambda} \bQ^T$, 
where $\bQ$ is an orthonormal eigenvector matrix and $\Lambda$ is the diagonal matrix of eigenvalues. Consequently 
\begin{equation*}
\bB\bB^T\bB\bB^T = (\bQ \pmb{\Lambda} \bQ^T)(\bQ \pmb{\Lambda} \bQ^T) = \bQ \pmb{\Lambda}^2 \bQ^T
\end{equation*}
Since $\pmb{\Lambda}^2$ squares the eigenvalues of $\pmb{\Lambda}$, we have $\bB\bB^T\bB\bB^T = K \bB\bB^T$.
\end{remark}

We continue with the proof of \textbf{Theorem \ref{Thm-clt-RKm}:}.

\medskip

\begin{proof}
Using $\widehat{\bmu}=\bN^+\bS$ and with some algebraic manipulation, we obtain
\begin{equation*}
R_{m,K}=\binom{K}{2}^{-1}\sum_{1\leq i<j\leq K} (\widehat{\nu}_{ij}-(\widehat{\mu }_{i}-\widehat{\mu }_{j}))^{2}= \binom{K}{2}^{-1}\pmb{V}_K^\top \pmb{V}_K,
\end{equation*}
where $\pmb{V}_K$ is a $\binom{K}{2}$ vector with $(i,j)^{th}$ entry $V_{ij}=\widehat{\nu}_{ij}-(\widehat{\mu }_{i}-\widehat{\mu }_{j})$ for $1\leq i<j\leq K$. Observe that $\bN^+=\frac{1}{K^2}\bN$ and $\bN=\bB^{\top}\bB$, where $\bB$ is the $\frac{K(K-1)}{2} \times K$ incidence matrix. Using $\bB\bB^T\bB\bB^T = K \bB\bB^T$ (see Remark \ref{remark:gof:1}), we have
\begin{align*}
\pmb{V}_K= \widehat{\bnu} - \bB \bN^{+}\bB^{\top}\widehat{\bnu}
= \widehat{\bnu} - \bB (\frac{1}{K^2}\bN)\bB^{\top}\widehat{\bnu}= \widehat{\bnu} - \frac{1}{K^2}\bB \bB^{\top}\bB \bB^{\top}\widehat{\bnu}=
\pmb{H}_K\widehat{\bnu},
\end{align*}
where $\pmb{H}_K= \pmb{I}-\frac{1}{K} \bB\bB^T$. Consequently, 
\[
R_{m,K} = \binom{K}{2}^{-1}\bV_K^\top  \bV_K= \widehat{\bnu}^{\top} \pmb{H}_K \widehat{\bnu}.
\]
Notice that for $1\leq i<j\leq K$
$${\rm Cov}(V_{ij},V_{kl})=\begin{cases}
    (K-2)\sigma^2/mK &\text{ if } (i,j)=(k,l),\\
    -\sigma^2/mK &\text{ if } (i,j)\neq (k,l) \text{ and } \{i,j\}\cap\{k,l\} \text{ is non empty},\\
    0 &\text{ if } (i,j)\neq (k,l) \text{ and } \{i,j\}\cap\{k,l\} \text{ is empty}.
\end{cases}$$ 

Next, $\widehat{\bnu}={\bnu}_{\rm linear}+{\bnu}_{\rm cyclic}+\pmb{\varepsilon}_K$, where $\pmb{\varepsilon}_K= (\varepsilon_{12},\ldots, {\varepsilon}_{K-1,K})$ and $\varepsilon_{ij} = \frac{1}{m} \sum_{k=1}^m \epsilon_{ijk}$. Therefore, using $\pmb{H}_K\,{\bnu}_{\rm linear}=\pmb{0}$, we get 
\[
R_{m,K} = \binom{K}{2}^{-1}({\bnu} + \pmb{\varepsilon}_K)^\top \pmb{H}_K ({\bnu} + \pmb{\varepsilon}_K) = \binom{K}{2}^{-1}{\bnu}_{\rm cyclic}^\top {\bnu}_{\rm cyclic} + 2\binom{K}{2}^{-1}\, {\bnu}_{\rm cyclic}^\top \pmb{\varepsilon}_K + \binom{K}{2}^{-1}\pmb{\varepsilon}_K^\top \pmb{H}_K \pmb{\varepsilon}_K.
\]
Using $\mathbb{E}[\varepsilon_{ij}]=0$ and expression for ${\rm Cov}(V_{ij},V_{kl})$ mentioned earlier, as $K\to\infty$ we have
$$
\binom{K}{2}^{-1}  {\bnu}_{\rm cyclic}^\top \pmb{\varepsilon}_K \xrightarrow[]{p} 0,\quad \text{and}\quad \binom{K}{2}^{-1} \pmb{\varepsilon}_K^\top \pmb{H}_K \pmb{\varepsilon}_K \xrightarrow[]{p} \sigma^2/m.
$$
Therefore, as $K\to\infty$,
\begin{align*}
R_{m,K} \xrightarrow[]{p}\lim_{K\to\infty}  \binom{K}{2}^{-1}\|{\bnu}_{\rm cyclic}\|^2\,+\, \sigma^2/m.
\end{align*}
This proves the first statement. 

Next, observe that
$\varepsilon_{ij}$'s, as defined above, are IID RVs with mean zero and variance $\sigma^2/m$. Furthermore, if $\epsilon_{ij}$'s have finite fourth moment then by Theorem 5.3 in \cite{DeJong1987} we have 
\[
\frac{R_{m,K} -   (\psi^2 + \sigma^2/m   )}{\sqrt{\operatorname{Var}(R_{m,K})}} \Rightarrow \mathcal{N}(0, 1),
\]
as $K\to\infty$. Next, when errors are normal, using the formula for the variance of the quadratic form of a normal random vector (see, \citealt{magnus1978}), we obtain ${\rm Var}(R_{m,K})=2\binom{K}{2}^{-2}\binom{K-1}{2}\frac{\sigma^4}{m^2} +4\,\binom{K}{2}^{-1}  \psi_K^2\,\frac{\sigma^2}{m}$. The approximate formula for the power function is obtained by using the limiting form of $R_{m,K}$ and setting $\psi_K^{2}=0$ in \eqref{eq:var:RK}. 
\end{proof}

\subsubsection*{Proof of Theorem \ref{thm:nij=1}:}
\begin{proof}
From Theorem \ref{Thm-clt-RKm}, as $K\to\infty$, $R_K\xrightarrow{p}\sigma^2$ and 
\begin{align}\label{eq:RJ:distr}
\frac{R_J-\sigma^2-\psi_J^2}{\sqrt{{\rm Var}({R_J})}}\Rightarrow \mathcal{N}(0,1).
\end{align}
Further, $J=o(K)$ allows to replace $\sigma^2$ by $R_K$ in \eqref{eq:RJ:distr}. The second part follows using steps similar to those in the proof of Theorem \ref{Thm-clt-RKm}.
\end{proof}

\subsubsection*{Proof of Theorem \ref{thm:random:graph}:}
\begin{proof}
Let $B_{\mathcal{V}}=\sum_{(i,j)\in\, \mathcal{V}}B_{ij}$, $B_{\mathcal{U}}=\sum_{(i,j)\in\, \mathcal{U}} B_{ij}$, and $\bD={\rm diag}(B_{1,2},\ldots,B_{K-1,K})$. Then, $B_{\mathcal{V}}/Kp_K\xrightarrow[]{p}1$ and  $B_{\mathcal{U}}/Jp_K\xrightarrow[]{p}1$. Using arguments similar to as in the proof of Theorem \ref{Thm-clt-RKm}, we get 
\[
\widetilde{R}_K = \frac{1}{B_{\mathcal{V}}}({\bnu} + \pmb{\varepsilon}_K)^\top \bD \pmb{H}_K \bD ({\bnu} + \pmb{\varepsilon}_K) = \frac{1}{B_{\mathcal{V}}} {\bnu}_{\rm cyclic}^\top \bD {\bnu}_{\rm cyclic} + 2\frac{1}{B_{\mathcal{V}}}\, {\bnu}_{\rm cyclic}^\top\bD \pmb{\varepsilon}_K + \frac{1}{B_{\mathcal{V}}}\pmb{\varepsilon}_K^\top \bD\pmb{H}_K \bD \pmb{\varepsilon}_K.
\]
Observe that, as $K\to\infty$,
$$
\frac{1}{B_{\mathcal{V}}}  {\bnu}_{\rm cyclic}^\top\bD \pmb{\varepsilon}_K \xrightarrow[]{p} 0,\quad  \frac{1}{B_{\mathcal{V}}}  \pmb{\varepsilon}_K^\top\bD \pmb{H}_K\, \bD\pmb{\varepsilon}_K \xrightarrow[]{p} \sigma^2 \quad \text{and}\quad
\frac{1}{B_{\mathcal{V}}} \sum_{(i,j)\in\,\mathcal{V}}B_{ij}{\nu}_{ij,\rm cyclic}^2 \xrightarrow[]{p}\, \psi_J^2.
$$
Moreover $\widetilde{R}_K$ can be rewritten as
\[
\widetilde{R}_K = \frac{1}{B_{\mathcal{V}}}(\bD({\bnu} + \pmb{\varepsilon}_K))^\top  \pmb{H}_K (\bD ({\bnu} + \pmb{\varepsilon}_K)).
\]
Consequently, if $\epsilon_{ij}$'s have finite fourth moment, using Theorem 5.3 in \cite{DeJong1987}, as $K\to\infty$ we get 
\begin{align*}
\frac{\widetilde{R}_J-\widetilde{R}_K-\widetilde{\psi}_J^2}{\sqrt{{\rm Var}({\widetilde{R}_J})}}\Rightarrow \mathcal{N}(0,1).
\end{align*}
\end{proof}

\end{document}